\documentclass[aip,jmp,nofootinbib,preprint]{revtex4-1}
\usepackage[a4paper, left=1.5cm, right=1.5cm, top=2.5cm, bottom=2.5cm, headsep=1.2cm]{geometry} 
\usepackage[utf8x]{inputenc}
\usepackage[T1]{fontenc}
\usepackage{graphicx}
\usepackage{enumerate}
\usepackage{float}
\usepackage{indentfirst}
\usepackage{amsmath}
\usepackage{amsfonts}
\usepackage{bbm}
\usepackage{amsthm}
\usepackage{multirow}
\usepackage{color}
\usepackage{pgf}
\usepackage{pgfplots}

\def\mk {\mathfrak}
\def\pl  {\delta}

\def\ZZ {{\mathbb Z}}
\def\RR {{\mathbb R}}

\def\D {\overline D}
\def\X {\overline X}
\def\A {\overline A}
\def\B {\overline B}
\def\obeta {\overline\beta}
\def\ker {{\rm ker}}
\def\im {{\rm im}}
\def\coker {{\rm coker}}
\def\coim {{\rm coim}}

\newtheorem{theorem}{Theorem}[section]
\newtheorem{lemma}[theorem]{Lemma}
\newtheorem{conj}[theorem]{Conjecture}

\newtheorem{coro}[theorem]{Corollary}
\newtheorem{rem}{Remark}[section]

\begin{document}

\title{
Homology groups for particles on one-connected graphs
}
\author{Tomasz Maci\c{a}\.{z}ek}
\affiliation{Center for Theoretical Physics, Polish Academy of
Sciences, Al. Lotnik\'ow 32/46, 02-668 Warszawa, Poland}
\author{Adam Sawicki}
\affiliation{Center for Theoretical Physics, Polish Academy of
Sciences, Al. Lotnik\'ow 32/46, 02-668 Warszawa, Poland}
\affiliation{School of Mathematics, University of Bristol, University Walk, Bristol BS8 1TW, UK}
\date{\today}

\begin{abstract}
We present a mathematical framework for describing the topology of configuration spaces for particles on one-connected graphs. In particular, we compute the homology groups over integers for different classes of one-connected graphs. Our approach is based on some fundamental combinatorial properties of the configuration spaces, Mayer-Vietoris sequences for different parts of configuration spaces and some limited use of discrete Morse theory. As one of the results, we derive a closed-form formulae for ranks of the homology groups for indistinguishable particles on tree graphs. We also give a detailed discussion of the second homology group of the configuration space of both distinguishable and indistinguishable particles. Our motivation is the search for new kinds of quantum statistics.
\end{abstract}

\maketitle

\section{Introduction}
The importance of the fundamental group of the configuration space $C_n(X)$ of $n$ indistinguishable particles living in a topological space $X$  to the description of quantum statistics was noted over 45 years ago \cite{L-M, Wilczek,Souriau,LD}.  The configuration space is defined as an orbit space $C_n(X) = (X^{\times n}-\Delta_n)/S_n$, where $\Delta_n$ corresponds to the coincident configurations, and $S_n$ is the permutation group. As was pointed by Dowker \cite{Dowker} abelian quantum statistics are classified by the first homology group, $H_1(C_n(X),\mathbb{Z})$, which is abelianization of the fundamental group, $\pi_1(C_n(X))$. In the topological approach, quantum statistics can be viewed as a connection with the vanishing curvature, i.e. there are no classical forces associated with it. The standard examples are when $X=\mathbb{R}^3$, $H_1(C_n(X),\mathbb{Z})=\mathbb{Z}_2$ and $X=\mathbb{R}^2$, $H_1(C_n(X),\mathbb{Z})=\mathbb{Z}$. They correspond to Bose/Fermi statistics in $\RR^3$ and {\it anyons} in $\RR^2$ (see \cite{Wilczek} for physical realisations of anyons). 

The significance of higher (co)homology groups for quantum theories is connected to classification of complex vector bundles. Recall for example that $U(1)$-vector bundles (complex line bundles) are classified by the first Chern class which is an element of $H^2(C_n(\mathbb{R}^2),\ZZ)$. The higher cohomology groups of the configuration spaces for particles in $\mathbb{R}^2$ were calculated by Arnold \cite{Arnold} in the 1960's. They have three basic properties: (1) finiteness: $H^{i}(C_{n}(\mathbb{R}^{2}))$ are finite except $H^{0}(C_{n}(\mathbb{R}^{2}))=\mathbb{Z}$, $H^{1}(C_{n}(\mathbb{R}^{2}))=\mathbb{Z}$ for $n\geq2$; also $H^{i}(C_{n}(\mathbb{R}^{2}))=0$ for $i\geq n$, (2) recurrence: $H^{i}(C_{2n+1}(\mathbb{R}^{2}))=H^{i}(C_{2n}(\mathbb{R}^{2}))$, (c) stabilization $H^{i}(C_{n}(\mathbb{R}^{2}))=H^{i}(C_{2i-2}(\mathbb{R}^{2}))$ for $n\geq2i-2$. In particular by property 1), excluding $H^0$ and $H^1$, they are purely torsions.  In the table provided in \cite{Arnold} we can see that $H^2(C_n(\mathbb{R}^2))=0$.  This implies that the only $U(1)$-vector bundle in $C_n(\mathbb{R}^2)$ is the trivial one. Therefore, for scalar particles one can infer that the full topological data is contained in $H_1(C_n(X))$. Interestingly, as was discussed in \cite{Bloore} for three particles in $\mathbb{R}^3$ one has $H^2(C_n(\RR^3))=\ZZ_2$ - we have two nonisomorphic $U(1)$-bundles corresponding to Bose/Fermi statistics. The authors of \cite{Bloore} also point out that for $SU(n)$-bundles, by looking at $H_4(C_3(\RR^3))=\mathbb{Z}_3$, one can see that there are three nonisomorphic bundles that correspond to some `symmetries' of a wavefunction $\psi(x_1,x_2,x_3)$ different form the usual bosonic or fermionic ones. The connections corresponding to these bundles are unfortunately not flat except for Bosons and Fermions that belong to the same classes. We note, however, that for compact $X$ the torsion part of $H^{2k}(C_n(X),\mathbb{Z})$ corresponds always to zero-curvature connection and therefore can be interpreted as a form of quantum statistics.

Recently there has been some interest in configuration spaces for both distinguishable  and indistinguishable particles on graphs. Also some models of interacting particles on graphs have been introduced, see \cite{Bolte} and the references therein. Here, by a graph we mean one dimensional cell complex. In \cite{HKR,HKRS,ASphd} spaces $C_n(\Gamma)$ for arbitrary connected graph $\Gamma$ were studied from the abelian quantum statistics perspective and the formula for $H_1(C_n(\Gamma),\mathbb{Z})$ was found. Essentially, abelian quantum statistics on graphs depends on the connectivity of a graph and its planarity. The only possible torsion is $\mathbb{Z}_2$ and it occurs when $\Gamma$ is a nonplanar graph. Interestingly higher homology groups of $C_n(\Gamma)$ are determined by $\pi(C_n(\Gamma))$ as the spaces $C_n(\Gamma)$ are aspheriacal (the so-called Eilenberg-MacLane space of type $K(\pi_1(C_n(\Gamma)),1)$). On the other hand, form Chern theory, we know they also provide some classification of possible gauge theories on $C_n(\Gamma)$. 

Understanding the structure of  $H_k(C_n(\Gamma))$, $k>1$ seems to be a first step towards understanding the role of higher homology groups in gauge theories over graph configuration spaces, and thus understanding possible new forms of quantum statistics. There is a limited number of results in this area. Homology groups of $C_n(\Gamma)$, where $\Gamma$ is a tree ($\Gamma=T$) have been studied from the Morse-theoretic point of view by Farley and Sabalka \cite{FStree}. The authors show that $H_k(D_n(T))$ are free with rank equal to the number of $k$-dimensional critical cells of the discrete vector field. However, as they point out, it is a difficult task to give a simple formula for the number of critical cells.

The key role in computing the homology groups is played by the $k$-dimensional cycles ($k$-cycles). A $k$-cycle is a subcomplex of the considered configuration space, which is a closed surface of dimension $k$. The elements of the homology group of order $k$ are represented by $k$-cycles, where the cycles that differ by a boundary of $(k+1)$-dimensional cells, are identified. In this paper, we construct an over-complete basis of $k$-cycles for indistinguishable particles on tree graphs, using the knowledge of the critical cells of the discrete Morse vector field \cite{FSbraid}. This approach allows us to find the closed-form formulae for the ranks of homology groups (section \ref{sec:tree}). The formulae involve only the first Betti numbers of configuration spaces for particles on star subgraphs of the tree, which are known \cite{HKRS,Ghirst}. The final result is given in equation (\ref{eq:hm_tree}). The main advantage of the mathematical framework we present in this paper is that it uses elementary features of the configuration spaces, i.e. explores the scheme of connections between different elements of the configuration space, which have a simple structure. In particular, we view a one-connected graph as a wedge of possibly higher-connected components (see Fig.\ref{components}). Next, we distinguish subspaces of the configuration space, that describe different distributions of particles between the components. Such subspaces have a structure of the cartesian product of configuration spaces of the components. Hence, by K\"{u}nneth theorem, the homology groups of the subspaces can be expressed by the homology groups of configuration spaces of the components. The most difficult task is to handle the relations between the homology groups that stem from the connections between the subspaces. The tool we use to describe these is Mayer-Vietoris sequences, a standard tool in the homology theory \cite{Hatcher}. For a more transparent presentation of our methodology, in section \ref{sec:configuration_spaces} we introduce the configuration space diagrams, which present the relevant parts of the configuration spaces. Other applications of our approach presented in this paper include the two-particle configuration spaces (section \ref{sec:two_particles}) and the case of two graphs connected by a single edge (section \ref{sec:one_edge}). There, we express the second Betti number of the configuration space by the first and second Betti numbers of the configuration spaces of components (formulas \ref{h2d2_indist}, \ref{h2d2_dist} and \ref{h2_2comp}). In this paper, we also partially consider the distinguished configuration spaces that play a major role in motion planning problems \cite{Farber03,FG08,Farber04,Farber05}. In particular,  in section \ref{sec:two_particles} we describe the two-particle case and in section \ref{sec:partial} we outline a possible procedure of generalising our description. Throughout the paper, we also discuss different difficulties that arise when trying to extend out methods to other classes of graphs.

\section{Preliminaries}\label{preliminaries}
\subsection{Configuration spaces as $CW$-complexes}
The configuration space for distinguishable and indistinguishable particles in a topological space $X$ is constructed in the following way. The configuration of $n$ particles in $X$ is represented by an element of $X^{\times n}$. 
We do not want two or more particles to occupy the same position, hence from the set of all configurations we subtract the diagonal $\Delta_n=\{(x_1,\dots,x_n)\in X^{\times n}:\ \exists_{i\neq j}\ x_i=x_j\}$. For indistinguishable particles we additionally identify configurations that differ by a permutation of particles, i.e. we take the quotient by the action of $S_n$, the permutation group. The definitions are as follows:
\[F_n(X):=X^{\times n}-\Delta_n,\ C_n(X):=(X^{\times n}-\Delta_n)/S_n.\]
We are particularly interested in the situation when $X=\Gamma$ is a graph. We will regard $\Gamma$ as a one-dimensional $CW$-complex. As an example, we show both configuration spaces for two particles on a $Y$-graph (Fig.\ref{y_configurations}).
 \begin{figure}[ht]
 
\includegraphics[width=0.7\textwidth]{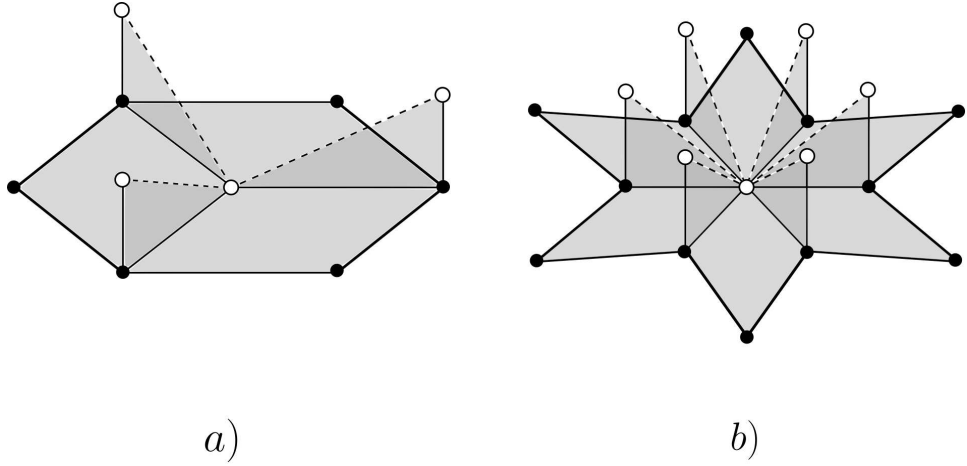}
\caption{Two-particle configuration space of a $Y$-graph embeded in $\RR^3$. a) $C_2(Y)$, b) $F_2(Y)$. The dashed lines and empty dots belong to the diagonal. For more examples of configuration spaces of graphs, see \cite{AbramsPhD,Ghirst}.}
\label{y_configurations}
\end{figure}
From figure \ref{y_configurations} we can see that spaces $F_n(\Gamma)$ and $C_n(\Gamma)$ do not have any simple structure. However, $F_n(\Gamma)$ and $C_n(\Gamma)$ can be deformation retracted to $CW$-complexes \cite{AbramsPhD} (Fig.\ref{y_complexes}). 
\begin{figure}[ht]
 
\includegraphics[width=0.7\textwidth]{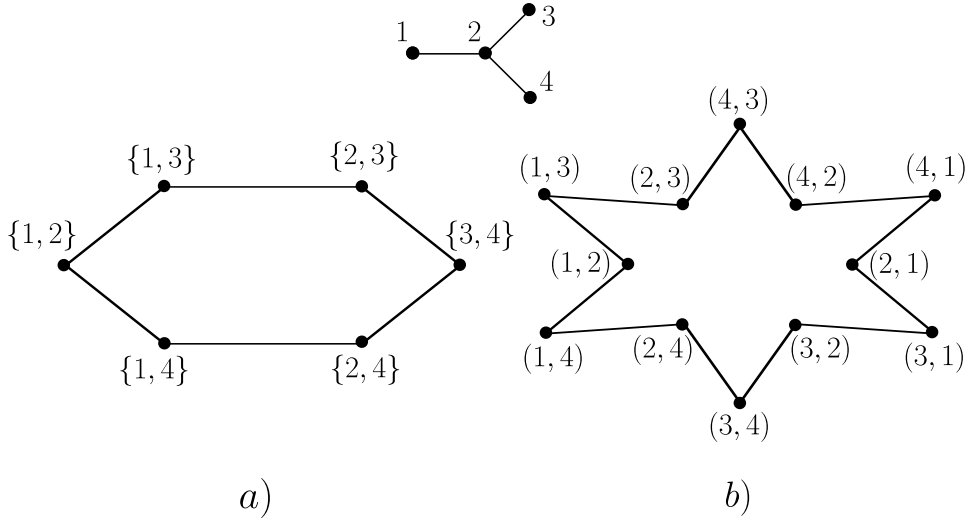}
\caption{Two-particle discrete configuration spaces of a $Y$-graph. Top picture is the $Y$-graph with ordered vertices, figures a) and b) show the discrete configuration spaces. a) $D_2(Y)$, b) $\D_2(Y)$.}
\label{y_complexes}
\end{figure}
For the deformation retraction to be valid, the graph must be {\it sufficiently subdivided} \cite{AbramsPhD}, which means that 
\begin{itemize}
\item each path between distinct vertices of degree not equal to 2 passes through at least $n-1$ edges,
\item each nontrivial loop passes through at least $n+1$ edges. 
\end{itemize}
In fact, the deformation retracts of the configuration spaces for graphs are cubical complexes. To see this, note first that $\Gamma^{\times n}$ already has the structure of a cubical complex. Namely, the $n$-dimensional cells are sequences of edges
\[\Sigma^{(n)}(\Gamma^{\times n})=\{(e_1,\dots,e_n):\ e_i\in E(\Gamma)\},\]
which are isomorphic to $n$-hypercubes. Two $n$-hypercubes in $\Gamma^n$ share a common face, i.e. an $(n-1)$-cell, when they are of the form
\[(e_1,\dots, e_{k-1}, e_k,e_{k+1},\dots,e_n),\ (e_1,\dots, e_{k-1}, e_k',e_{k+1},\dots,e_n),\ e_k{\rm\ adjacent\ to}\ e_k'.\]
 Then, the shared face is $(e_1,\dots, e_{k-1}, v,e_{k+1},\dots,e_n)$, where $v=e_k\cap e_k'$. The $n-1$-dimensional cells are sequences of $n-1$ edges and a vertex from $\Gamma$, and so on. The cubical complex, which is the deformation retract of a configuration space is called the {\it discrete configuration space}. The discrete configuration space of $F_n(\Gamma)$ will be denoted by $\D_n(\Gamma)$, while the discrete version of $C_n(\Gamma)$ by $D_n(\Gamma)$. $\D_n(\Gamma)$ or $D_n(\Gamma)$ are complexes, whose $n$-dimensional skeletons are composed of cells of the following form
\begin{eqnarray*}
\Sigma^{(n)}(\D_n(\Gamma))=\{(e_1,\dots,e_n):\ e_i\in E(\Gamma),\ e_i\cap e_j=\emptyset\ {\rm\ for\ all}\ i,j\}, \\
\Sigma^{(n)}(D_n(\Gamma))=\{\{e_1,\dots,e_n\}:\ e_i\in E(\Gamma),\ e_i\cap e_j=\emptyset\ {\rm\ for\ all}\ i,j\}.
\end{eqnarray*}
Lower dimensional cells are described by sequences or collections of edges and vertices from $\Gamma$. A $k$-dimensional cell contains $k$ edges and $n-k$ vertices. In other words, cells from $\Sigma^{(k)}(\D_n(\Gamma))$ are of the form
\begin{eqnarray*}
\sigma=(\sigma_1,\dots,\sigma_n):\ \sigma_i\cap \sigma_j=\emptyset\ {\rm\ for\ all}\ i,j,\ {\rm and}\ \#(\sigma\cap E(\Gamma))=k,\ \#(\sigma\cap V(\Gamma))=n-k.
\end{eqnarray*}
Similarly, cells of the $k$-skeleton of $D_n(\Gamma)$ are
\begin{eqnarray*}
\sigma=\{\sigma_1,\dots,\sigma_n\}:\ \sigma_i\cap \sigma_j=\emptyset\ {\rm\ for\ all}\ i,j,\ {\rm and}\ \#(\sigma\cap E(\Gamma))=k,\ \#(\sigma\cap V(\Gamma))=n-k.
\end{eqnarray*}
In particular, when there are not enough pairwise disjoint edges in the sufficiently subdivided $\Gamma$, the dimension of the discrete configuration space can be less than $n$.

Another important notion are the {\it $k$-dimensional chains} ($k$-chains) from the discrete configuration space. A $k$-chain in $\D_n(\Gamma)$ or $D_n(\Gamma)$ is a formal linear combination of $k$-cells, whose coefficients are integers.
\[\mk{C}_k=\left\{\sum_{\sigma\in\Sigma^{(k)}}a_\sigma\sigma:\ a_\sigma\in\ZZ\right\}.\]
Next, we define a boundary map that maps $k$-chains to $k-1$-chains and that satisfies $\partial\partial=0$. The boundary map is defined on $k$-cells and extends by linearity to $k$-chains. Having chosen a spanning tree of the sufficiently subdivided graph $T\in\Gamma$ and its plane embedding, we define the boundary of each edge by the following procedure of numbering the vertices \cite{KoPark,FSbraid}. The root of $T$ has number $1$ and we move along the tree to number the remaining vertices. If a vertex has degree $\geq 3$, we number the vertices in each branch in the clockwise order. Then, each edge $e\in E(\Gamma)$ has its initial and terminal vertex $\iota(e)>\tau(e)$ and the boundary of $e$ is the following $0$-chain:
\[\partial e=\iota(e)-\tau(e).\]
Similarly, for every $k$-cell we have $k$ pairs of initial and terminal faces. In the case of a cell from $D_n(\Gamma)$
\[\sigma=\{e_1,\dots,e_k,v_1,\dots,v_{n-k}\}, \]
we additionally order the edges from $\sigma$ according to their terminal vertices, i.e. $\tau(e_1)<\tau(e_2)<\dots<\tau(e_k)$. The $i$th pair of faces from the boundary of $\sigma$ reads
\begin{eqnarray*}
\left(\partial^\iota\sigma\right)_i:=\{e_1,\dots,e_{i-1},e_{i+1},\dots,e_k,v_1,\dots,v_{n-k},\iota(e_i)\}, \\
\left(\partial^\tau\sigma\right)_i:=\{e_1,\dots,e_{i-1},e_{i+1},\dots,e_k,v_1,\dots,v_{n-k},\tau(e_i)\}.
\end{eqnarray*}
The full boundary of $\sigma$ is given by the following alternating sum of faces.
\begin{equation}\label{eq:boundary}
\partial\sigma=\sum_{i=1}^k(-1)^k\left(\left(\partial^\iota\sigma\right)_i-\left(\partial^\tau\sigma\right)_i\right).
\end{equation}
An analogous formula holds for $\D_n(\Gamma)$. Then, index $i$ numbers only the edges in $\sigma$.

We are interested in the description of those chains, whose boundary is empty. Such chains of dimension $p$ are called $p$-cycles. Denote the boundary map acting on the set of $p$-chains by $\partial_p$. Having chosen a basis of $p$-cycles, we have $\ker \partial_p\simeq\ZZ^d$, where $d$ is the number of basis elements. Moreover, we want to identify those cycles, that belong to the boundary of some $p+1$-chain. Such objects are elements of the homology group of order $p$ \cite{Hatcher}.
\[H_p(C_n(\Gamma),\ZZ)=\ker\partial_p/\im\partial_{p+1}.\]
Hence, $H_p$ is a finitely generated abelian group. By the structure theorem, we know that $H_p$ has a free part (sum of copies of $\ZZ$) and a torsion-part. The quotient by ${\rm Im}\partial_{k+1}$ means that the $k$-cycles, that belong to the boundary of the same $(k+1)$-cell, are identified.

Let us next review the known results regarding the first homology group of the configuration spaces of indistinguishable particles, which will also serve as an instructive example for the above theory.  The space of $1$-cycles and the relations between the elements of the spanning set are known for any simple graph \cite{HKRS}.
\begin{theorem}[The spanning set of $\ker\partial_1$ \cite{HKRS}.]
For $D_n(\Gamma)$, where $\Gamma$ is any simple graph, an over-complete basis of $1$-cycles can be constructed from the following elements. 
\begin{enumerate}[i)]
\item One particle travelling along one cycle in $\Gamma$ and the remaining particles being fixed on the vertices, which are disjoint with the cycle.
\item Two particles exchanging on a $Y$-subgraph (fig. \ref{Y_cycle}), while the remaining particles being fixed on the vertices, which are disjoint with the $Y$-subgraph.
\end{enumerate}
\end{theorem}
\noindent Consider the simplest one-particle $O$-cycle, where a particle moves along a cycle, which consists of three edges. Such a cycle can be written formally as the following linear combination of $1$-cells:
\begin{equation}\label{O_cycle}
c_O=\{e_1^2\}+\{e_2^3\}-\{e_1^3\}.
\end{equation}
It is straightforward to check that $\pl_1(c_O)=0$ for the boundary map from (\ref{eq:boundary}). Another type of one-particle cycle is an exchange of particles on a $Y$-subgraph. The two-particle configuration space of this cycle is shown on Fig.\ref{y_complexes}a. 
\noindent The cycle, which we will refer to as $c_Y$-cycle is 
\begin{equation}\label{Y_cycle}
c_Y=\{e_{2}^{3},1\}+\{e_{1}^{2},3\}+\{e_{2}^{4},3\}-\{e_{2}^{3},4\}-\{e_{1}^{2},4\}-\{e_{2}^{4},1\}.
\end{equation}
\noindent For any pair $1$-chains in $\Gamma$,
\[c=\sum_{i}a_i\{e_i\},\ c'=\sum_{j}b_j\{e'_j\},\] 
one can construct their tensor product, which is given by the following formula
\[c \otimes c'=\sum_{i,j}a_i b_j\{e_i\}\otimes\{e'_j\}.\]
If the chains are disjoint, i.e. $e_i\cap e_j'=\emptyset$ for all $i,j$, the above tensor product can be embedded into $D_2(\Gamma)$ by taking $\{e_i\}\otimes\{e'_j\}\mapsto \{e_i,e_j\}$. By choosing $c$ and $c'$ to be disjoint $1$-cycles, the resulting $2$-chain is a $2$-cycle, which is topologically a cartesian product of circles, i.e. a $2$-torus. Note that by taking the tensor product of a sufficiently large number of disjoint $1$-cycles, one can construct larger dimensional tori. The part of $\ker\partial_p$, which has such a form will be referred to as the {\it toric} part. For example, the toric part of $\ker\partial_2$ in $D_2(\Gamma)$ is given by $c_O\otimes c_O'$, where $c_O$ and $c_O'$ are disjoint $1$-cycles from $\Gamma$. The toric part of $\ker\partial_2$ in $D_3(\Gamma)$ is $c_O\otimes c_O'\otimes\{v\}$ with $c_O\cap c_O'\cap v=\emptyset$, and $c_O\otimes c_Y$, $c_O\cap c_Y=\emptyset$. As we show in the next paragraphs, a key role in the computation of homology groups for particles on tree graphs is played by tori, which are products of disjoint $c_Y$-cycles.

In this paper, we will also partially use Forman's Morse theory for $CW$-complexes \cite{Forman}, which has been formulated for discrete graph configuration spaces in \cite{KoPark,FSbraid,FStree,AS12}. We do not review all the details of this construction, since it is just a background for the methods that are developed in this paper. We restrict ourselves to a brief description of the general idea standing behind the discrete Morse theory. 

The discrete Morse theory is a construction of a homotopy deformation of a cell complex, which collapses some of the cells, effectively shrinking the complex. The collapse is performed along the flow of the discrete gradient vector field. The discrete gradient vector field is a map from $k$-cells to $(k+1)$-cells that satisfies certain conditions \cite{Forman}. Cells, that are in the image of the discrete vector field are ${\it collapsible}$. Cells from the domain of the vector field are ${\it redundant}$. The cells of both these kinds are collapsed by the deformation induced by the flow of the vector field. The cells of the third kind are the cells, which are neither in the domain nor in the image of the vector field, i.e. the ${\it critical}$ cells. Such cells constitute the Morse complex. In the following, we will often use the knowledge of the critical cells to construct an over-complete basis of cycles.

\subsection{Presentation of results for tree graphs}
The main contribution of this work is the computation of the homology groups for indistinguishable particles on tree graphs (section \ref{sec:tree}). The methods that are used in this context, can be also extended to the case of two particles on a $1$-connected graph (section \ref{sec:two_particles}) and two graphs connected by a single edge (section \ref{sec:one_edge}). We will regard a tree graph as a loopless lattice of star graphs. A star graph is a graph with a single vertex of degree larger than $2$ (called the hub, or the central vertex) and a number of edges attached to the vertex. Namely, for every tree graph one can construct the underlying tree, whose vertices denote the central vertices of the star graphs and the edges symbolise the connections between the star graphs, see Fig.\ref{tree}. 
 \begin{figure}[ht]
 
\includegraphics[width=0.5\textwidth]{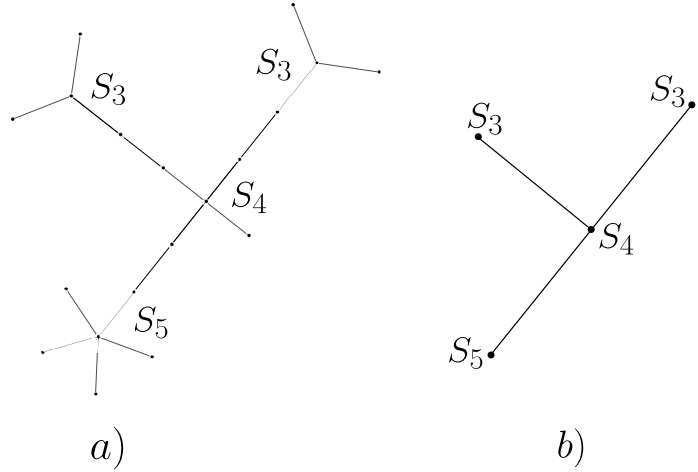}
\caption{ a) A tree graph regarded as a lattice of star graphs. b) The underlying scheme of connections.}
\label{tree}
\end{figure} 
 The homology groups for star graphs are well-known. In particular, $D_n(S)$ is homotopy equivalent to a wedge sum of circles \cite{Ghirst,AbramsPhD} \footnote{A wedge sum of topological spaces $X$ and $Y$ is a space, which is created by identifying a point in $X$ with a point in $Y$. In other words, this is the space $(X\sqcup Y)/\sim$, where $\sim$ is the quotient map that identifies the two distinguished points.}. Hence, $H_k(D_n(S))=0$ for $k\geq2$. Moreover, the dimension of $H_1(D_n(S))$ is given by \cite{HKRS}
\begin{equation}\label{h1star_ind}
\beta_1^{(n)}(S)={{n+E-2}\choose{E-1}}(E-2)-{{n+E-2}\choose{E-2}}+1,
\end{equation}
where $E$ is the number of edges adjacent to the central vertex of $S$. The same is true for $\D_n(S)$, except the number of generators of $H_1(\D_n(S))$ reads \cite{Ghirst}
\begin{equation}\label{h1star_dist}
\obeta_1^{(n)}(S)=1+(nE-2n-E+1)\frac{(n+E-2)!}{(E-1)!}.
\end{equation}
Both formulae are connected by Euler characteristics of the complexes. Namely, because each cell of $D_n(S)$ is covered $n!$ times in $\D_n(S)$, we have $\overline\chi=n!\chi$, which implies that $1-\obeta_1^{(n)}(S)=n!(1-\beta_1^{(n)}(S))$.

Homology groups of $D_n(T)$ have been studied from the Morse-theoretic point of view by Farley and Sabalka in \cite{FStree}. The authors show that $H_k(D_n(T))$ are free with rank equal to the number of $k$-dimensional critical cells of the discrete vector field. However, as they point out, it is a difficult task to give a simple formula for the number of critical cells. In this paper, we construct an over-complete basis of $\partial_k$ for $D_n(T)$ using the knowledge of the critical cells of the discrete vector field. This approach allows us to find closed-form formulae for the numbers of the critical cells. The idea is to construct a $k$-cycle for a given critical $k$-cell, which contains the critical cell and which is carried by the vector field's flow to the corresponding cell in the Morse complex. Critical cells of the discrete vector field for trees are known \cite{KoPark,FStree}. A critical cell contains $k$ disjoint edges and $n-k$ distinct vertices. Each edge from the considered critical $k$-cell contains a vertex of degree $\geq 3$. Let us call such a vertex the {\it hub} of a star graph. The vertices from the critical cell are {\it blocked}, i.e. are stacked behind the hubs or stacked behind the tree's root\footnote{For a full description of critical cells of the discrete vector field, see \cite{KoPark,FStree}.}. See Fig.\ref{critical_example} for an example of a critical $2$-cell. 
\begin{figure}[ht]
 
\includegraphics[width=0.4\textwidth]{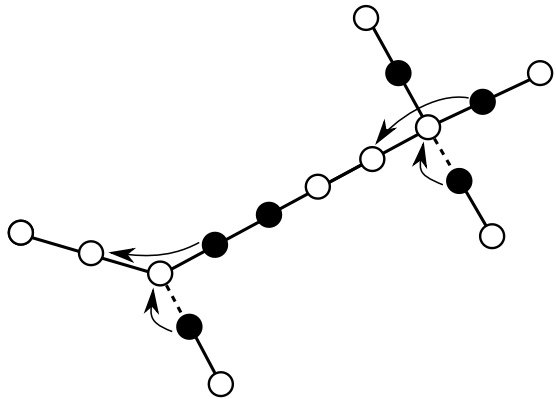}
\caption{The correspondence between the critical cells of the discrete gradient vector field and cycles in the configuration space. The edges from the critical cell are marked with dashed lines. Arrows mark the $Y$-subgraphs, where the pairs of particles exchange. The occupied vertices (black dots) denote the free particles, which are stacked behind the hubs.}
\label{critical_example} 
\end{figure}
The corresponding cycle, that is carried by the vector field's flow to a proper critical cell is of the form
\[c_{Y_1}\otimes c_{Y_2}\otimes\dots\otimes c_{Y_k}\otimes\{v_1,\dots,v_{n-2k}\},\]
where each of the $Y$-subgraphs consists of one edge from the critical cell, one edge, which contains the hub and a free vertex in the star graph, and one edge, which contains the hub and a vertex from the critical cell (Fig.\ref{critical_example}). Clearly, such a cycle contains a single critical cell, which is the desired one. Moreover, one can check that the remaining cells from such a cycle are collapsible or redundant, i.e. are collapsed by the vector field's flow.  The over-complete basis for tree graphs that we consider, consists of all such cycles, however we do not require the vertices $\{v_1,\dots,v_{n-2k}\}$ to be blocked. They can be arbitrary vertices of $T$. Hence, our over-complete basis is much larger than the number of critical cells. Section \ref{sec:tree} uses a proper topological machinery to handle the relations between the cycles from the over-complete basis. These relations come from the $(k+1)$-cells from $D_n(T)$ and from linear dependence within $\ker\partial_p$.

The fact that the homology groups are free, can be also proved using the above correspondence between the critical cells and cycles in $D_n(T)$. Namely, there are no relations between the $k$-cycles of the Morse complex that stem from the boundaries of $(k+1)$-cells. This is because every cell of the Morse complex has no boundary, as it is the image under the deformation retraction of a cycle in the configuration space. Such a deformation maps cycles to cycles.

Let us next describe the intuition standing behind the proof of the formulae for the ranks of the homology groups. Particles on a tree graph can exchange only on $Y$-subgraphs. Using the above arguments from Morse theory, it is enough to consider exchanges of pairs of particles that involve separate $Y$-subgraphs. There are two kinds of relations between the cycles corresponding to such exchanges
\begin{itemize}
\item relations between the exchanges on $Y$-subgraphs from the same star subgraph,
\item relations between the exchanges on distinct star subgraphs, that stem from the connections between the subgraphs.
\end{itemize}
The relations of the first kind can be handled by choosing the $1$-cycles, which are the representatives of the basis of the first homology group for the proper star subgraphs. The number of independent $1$-cycles for particles on a star graph is given by formula (\ref{h1star_ind}). For example, consider a tree, which consists of exactly two star graphs (Fig.\ref{tree-components}), $S$ and $S'$. The representatives of the second homology group for four particles on such a tree are the $2$-cycles, which are products of two-particle $1$-cycles from $S$ and two-particle $1$-cycles from $S'$. Hence, the rank of $H_2(D_4(S,S'))$ reads
\[\beta_2^{(4)}(S,S')=\beta_1^{(2)}(S)\beta_1^{(2)}(S').\] 
When the number of particles on $T$ is larger than $4$, the relations of the second kind come into play. There are $2$-cycles, that come from all possible distributions of particles between $S$ and $S'$. In the case of $5$ particles, the number of such cycles is $\beta_1^{(2)}(S)\beta_1^{(3)}(S')+\beta_1^{(3)}(S)\beta_1^{(2)}(S')$. However, each such $2$-cycle has one free particle, that does not take part in the exchange. Hence, the $2$-cycles, where the free particle is sitting on the path connecting $S$ and $S'$, where counted twice. To obtain the rank of $H_2(D_5(S,S'))$, we have to subtract the double-counted cycles, whose number is $\beta_1^{(2)}(S)\beta_1^{(2)}(S')$. In the case of $n$ particles, we have to subtract the cycles, where at least one particle is sitting on the connecting path, which is exactly the number of (over-complete) cycles for $n-1$ particles. The formula reads
\[\beta_2^{(n)}(S,S')=\sum_{l=2}^{n-2}\left(\beta_1^{(l)}(S)-\beta_1^{(l-1)}(S)\right)\beta_1^{(n-l)}(S').\]
A similar result holds for a situation, where star graph $S$ is connected with a single edge to a tree graph $T'$, i.e. $S'$ can be replaced in the above formula by $T'$. Recall that $\beta_1^{(n-l)}(T')$ is the sum of $\beta_1^{(n-l)}(S')$ for all $S'\subset T'$. Hence, we have $\beta_2^{(n)}(T)=\sum_{(S,S')\subset T}\beta_2^{(n)}(S,S')$ for any tree. The same reasoning can be used to compute the rank of the $k$th homology group. The conditions for $H_m$ to be nonzero are: i) the tree contains at least $m$ star subgraphs, ii) the number of particles is at least $2k$. The simplest case is when $n=2m$ and the tree contains exactly $m$ star subgraphs. Then, it is enough to multiply the two-particle $1$-cycles from the distinct star subgraphs, i.e.
\[\beta_m^{(2m)}(T)=\prod_{S\subset T}\beta_1^{(2)}(S)\ {\rm for}\ \#T=m.\]
For a larger number of particles, handling the multiply-counted cycles is a more difficult task than in the case of $H_2$, because there are more connecting paths, where the free particles can be distributed. However, this problem can be tackled recursively. Consider a star graph $S$ connected by an edge with a tree, which consists of $m-1$ star subgraphs. Every $m$-cycle in $H_m(D_n(S,T'))$, $n>2m$, is a product of a $1$-cycle from $D_l(S)$ and a  $(m-1)$-cycle from $D_{n-l}(T')$. Multiplying the $(m-1)$-cycles with the $1$-cycles and subtracting the multiply-counted cycles, we get
\[\beta_m^{(n)}(S,T')=\sum_{l=2}^{n-2}\left(\beta_1^{(l)}(S)-\beta_1^{(l-1)}(S)\right)\beta_{m-1}^{(n-l)}(T').\]
Considering tree $T'$ as a star $S'$ connected by an edge with tree $T''$, we get a similar relation for $\beta_{m-1}^{(n-l)}(T')$. Proceeding in this way, we end up with a formula, which expresses $\beta_m^{(n)}(T)$ by the first Betti numbers of the star subgraphs of $T$ for different distributions of particles. The final expression is given in equation (\ref{eq:hm_tree}) in section \ref{sec:tree}. 

The above reasoning is just a sketch of the main ideas standing behind the rigorous proof, which is given in the following sections of this paper. 

\section{Configuration space for one-connected graphs}\label{sec:configuration_spaces}
In this section, we describe the structure of $D_n(\Gamma)$ and $\D_n(\Gamma)$ for one-connected graphs. The results of this section play a key role in the method of computing homology groups by Mayer-Vietoris sequences. As a preliminary exercise, consider the case of $n$ particles on a disjoint sum of two graphs, $\Gamma=\Gamma_1\sqcup \Gamma_2$. One can distinguish different parts of $D_n(\Gamma)$, which correspond to distributing $k$ particles on $\Gamma_1$ and $l$ particles on $\Gamma_2$, $k+l=n$. Because the two graphs are disjoint, such a component is isomorphic to the Cartesian product of the corresponding configuration spaces for $\Gamma_1$ and $\Gamma_2$, i.e. $D_k(\Gamma_1)\times D_l(\Gamma_2)$. The following lemma shows, that there are no connections between different components.

 \begin{lemma}
 The $n$-particle configuration space for a disjoint pair of graphs is a disjoint union of the following connected components.
 \[D_n(\Gamma_1\sqcup\Gamma_2)=\bigsqcup_{k+l=n}D_{k}(\Gamma_1)\times D_l(\Gamma_2).\]
 \end{lemma}
 \begin{proof}
Because $\Gamma_1$ and $\Gamma_2$ are disjoint, there is no possibility for the particles to move from one graph to another. Such a possibility is essential for the existence of connections between different components of $D_n(\Gamma_1\sqcup\Gamma_2)$. More formally, for any two cells
 \[c\in D_{k}(\Gamma_1)\times D_l(\Gamma_2)\ {\rm and}\ c'\in D_{k'}(\Gamma_1)\times D_{l'}(\Gamma_2),\]
where $k\neq k'$ or $l\neq l'$, there is no path in $D_n(\Gamma_1\sqcup \Gamma_2)$ joining $c$ with $c'$. To see this, recall that the $n$-particle configuration space is a cubic complex. The existence of a path joining two vertices of a cubic complex, is equivalent to the existence of a $1$-chain in the complex that joins the two vertices. Such a $1$-chain necessarily contains a $1$-cell, whose endpoints belong to $D_{k}(\Gamma_1)\times D_l(\Gamma_2)$ and $D_{k'}(\Gamma_1)\times D_{l'}(\Gamma_2)$ respectively. We will next show that such a $1$-cell does not exist. The endpoints of such a cell are the $0$-cells of the form
\begin{eqnarray*}
c^{(0)}=\{v_1^{(1)},v_2^{(1)},\dots,v_k^{(1)},v_1^{(2)},v_2^{(2)},\dots,v_l^{(2)}\}\in D_{k}(\Gamma_1)\times D_l(\Gamma_2), \\
{c'}^{(0)}=\{{v'}_1^{(1)},{v'}_2^{(1)},\dots,{v'}_{k'}^{(1)},{v'}_1^{(2)},{v'}_2^{(2)},\dots,{v'}_{l'}^{(2)}\}\in D_{k'}(\Gamma_1)\times D_{l'}(\Gamma_2),
\end{eqnarray*}
where $v_i^{(1)},{v'}_i^{(1)}\in \Gamma_1,\ v_i^{(2)},{v'}_i^{(2)}\in \Gamma_2$. For ${c'}^{(0)}$ and ${c}^{(0)}$ to be the endpoints of a $1$-cell, there must exist a pair of vertices $(v_i^{(1)},{v'}_j^{(1)})$ or $(v_i^{(2)},{v'}_j^{(2)})$, who are adjacent in $\Gamma_1$ or $\Gamma_2$ respectively. Without loss of generality, we can assume that $(v_1^{(1)},{v'}_1^{(1)})$ is such a pair. Then, any $1$-cell that contains $c^{(0)}$ is of the form
\[c^{(1)}=\{e,v_2^{(1)},\dots,v_k^{(1)},v_1^{(2)},v_2^{(2)},\dots,v_l^{(2)}\},\]
where $\partial_1(e)=\pm(v_1^{(1)}-{v'}_1^{(1)})$. Therefore, both endpoints of any $1$-cell containing $c^{(0)}$ belong to $D_{k}(\Gamma_1)\times D_l(\Gamma_2)$, which is a contradiction.
 \end{proof}
An analogous result holds for $\D_n(\Gamma_1\sqcup\Gamma_2)$, where the components $\D_{k}(\Gamma_1)\times \D_l(\Gamma_2)$ split into $n\choose k$ parts with different permutations of particles. 
Because the configuration is a disjoint sum, we have
\[H_m(D_n(\Gamma_1\sqcup\Gamma_2))=\bigoplus_{k+l=n}H_m(D_{k}(\Gamma_1)\times D_l(\Gamma_2)).\]
Furthermore, by K\"unneth theorem
\[H_m(D_{k}(\Gamma_1)\times D_l(\Gamma_2))=\bigoplus_{i+j=m}H_i(D_{k}(\Gamma_1))\otimes H_j(D_{l}(\Gamma_2)).\]
Hence, for a disjoint sum of graphs, the knowledge of the homology groups of the configuration spaces of the components is sufficient to compute the homology groups of $D_n(\Gamma)$. Let us next move to the case of one-connected graphs\footnote{From now on, $\Gamma$ will always denote a one-connected graph.}. A one-connected graph is a graph, which becomes disconnected after removing a particular vertex together with the adjacent edges. Note that $D_n(\Gamma)$, where $\Gamma$ is a one-connected graph, can be treated in a similar way by considering components that arise from splitting the graph at a proper vertex (see Fig.\ref{splitting}). 
 \begin{figure}[ht]
 
\includegraphics[width=0.5\textwidth]{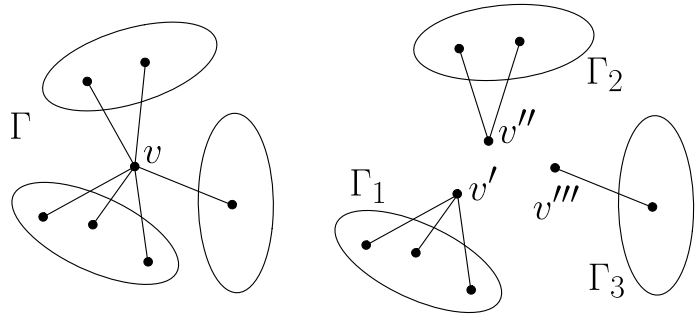}
\caption{A one-connected graph $\Gamma$ as a wedge sum of three components. $\Gamma=(\sqcup_i\Gamma_i)/\sim$, where $v'\sim v'',\ v''\sim v'''$.}
\label{splitting}
\end{figure} 
In other words, any one-connected graph can be viewed as a wedge sum of graphs, which we call components. Consider first a simpler case, where $\Gamma$ has two components. Our goal is to describe the connections in $D_n(\Gamma)$ between the components $D_{k}(\Gamma_1)\times D_l(\Gamma_2)$ that are induced by the gluing map. In fact, we have to consider the disjoint subgraphs of $\Gamma$, hence sometimes we have to remove vertex $v$ from each component. To this end, we do an extra subdivision of edges that connect $\Gamma_i$ with $v$ and remove the last segment of each such edge. The component after such an operation will be denoted by $\tilde\Gamma_i$ (see Fig.\ref{components}).
 \begin{figure}[ht]
 
\includegraphics[width=0.7\textwidth]{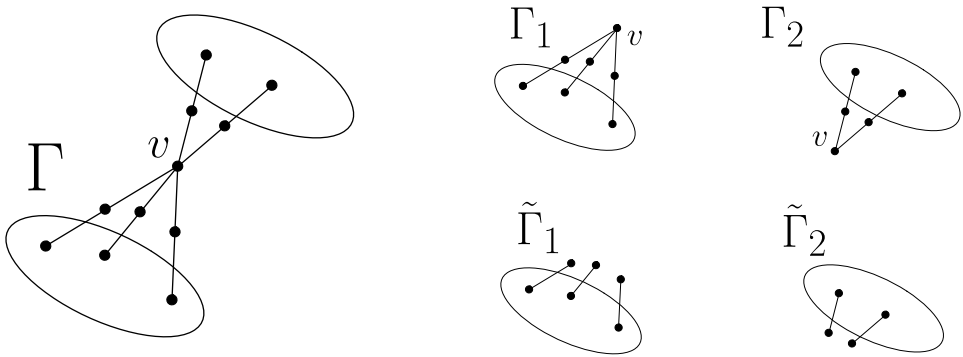}
\caption{The components of a $1$-connected graph with two components.}
\label{components}
\end{figure} 
Let us next show how the different components of $\Gamma$ come into play in $D_2(\Gamma)$. Cells of $D_2(\Gamma)$ are
\begin{eqnarray*}
\Sigma^{(0)}(D_2(\Gamma))=\{\{v,v'\}:\ v\neq v'\},\ \Sigma^{(1)}(D_2(\Gamma))=\{\{e,v\}:\ e\cap v=\emptyset\},\\ \Sigma^{(2)}(D_2(\Gamma))=\{\{e,e'\}:\ e\cap e'=\emptyset\}.
\end{eqnarray*}
Next, we write each set of cells as a sum of cells from different components.
\begin{eqnarray*}
\Sigma^{(0)}(D_2(\Gamma))=\{\{v,v'\}:\ v\neq v'\ {\rm and}\ v,v'\in V(\Gamma_1)\}\cup \{\{v,v'\}:\ v\neq v'\ {\rm and}\ v,v'\in V(\Gamma_2)\}\cup \\
\cup \{\{v,v'\}:\ v\in V(\Gamma_1)\ {\rm and}\  v'\in V(\tilde\Gamma_2)\}\cup \{\{v,v'\}:\ v\in V(\tilde\Gamma_1)\ {\rm and}\  v'\in V(\Gamma_2)\}.
\end{eqnarray*}
The sets of cells of a higher dimension can be written in an analogous way. In other words, 
\[\Sigma^{(i)}(D_2(\Gamma))=\Sigma^{(i)}(D_2(\Gamma_1))\cup\Sigma^{(i)}(D_2(\Gamma_2))\cup\Sigma^{(i)}(\Gamma_1\times\tilde\Gamma_2)\cup\Sigma^{(i)}(\tilde\Gamma_1\times\Gamma_2).\]
Some of the above summands are not disjoint, i.e.
\begin{eqnarray*}
\Sigma^{(i)}(D_2(\Gamma_1))\cap\Sigma^{(i)}(\tilde\Gamma_1\times\Gamma_2)=\Sigma^{(i)}(\tilde\Gamma_1\times v),\ \Sigma^{(i)}(\tilde\Gamma_1\times\Gamma_2)\cap\Sigma^{(i)}(\Gamma_1\times\tilde\Gamma_2)=\Sigma^{(i)}(\tilde\Gamma_1\times\tilde\Gamma_2), \\ \Sigma^{(i)}(D_2(\Gamma_2))\cap\Sigma^{(i)}(\Gamma_1\times\tilde\Gamma_2)=\Sigma^{(i)}(\tilde\Gamma_2\times v).
\end{eqnarray*}
The structure of $D_2(\Gamma)$ is shown on Fig.\ref{D2G}.
 \begin{figure}[ht]
\includegraphics[width=0.6\textwidth]{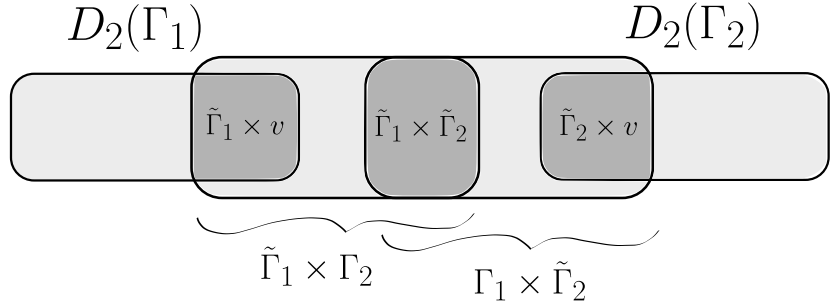}
\caption{A scheme of $D_2(\Gamma)$ for $\Gamma$ with two components.}
\label{D2G}
\end{figure} 
Configuration space from Fig.\ref{D2G} can be represented as a diagram, Fig.\ref{D2Gdiag}. A node represents a subcomplex of $D_2(\Gamma)$, while the edges describe the common parts of neighbouring subcomplexes.
 \begin{figure}[ht]
\includegraphics[width=0.8\textwidth]{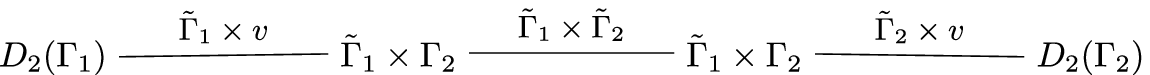}
\caption{Configuration space diagram of $D_2(\Gamma)$ for $\Gamma$ from Fig.\ref{components}.}
\label{D2Gdiag}
\end{figure} 
\noindent The configuration space for distinguishable particles has a similar stucture, except one has to take into account different permutations of particles -- each component of $D_2(\Gamma)$ is covered twice, hence the diagram has two branches, see Fig.\ref{D2Gdiag_dist}.
 \begin{figure}[ht]
\includegraphics[width=0.8\textwidth]{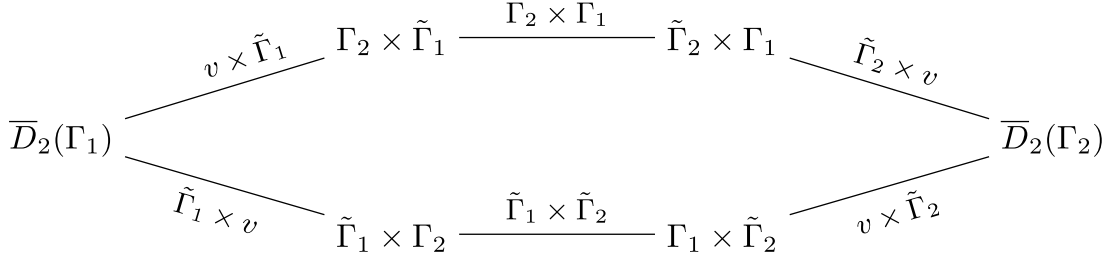}
\caption{Configuration space diagram of $\D_2(\Gamma)$ for $\Gamma$ from Fig.\ref{components}.}
\label{D2Gdiag_dist}
\end{figure} 
\noindent For $n>2$, one has many possibilities of distributing the particles among the components. However, the structure of $D_n(\Gamma)$ is still linear, in the sense that it is a chain of complexes, Fig.\ref{DnGdiag}.
 \begin{figure}[ht]
\includegraphics[width=\textwidth]{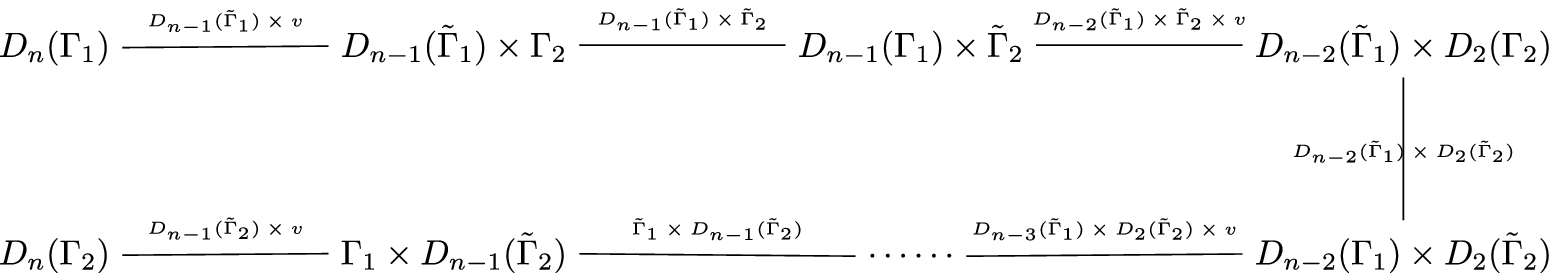}
\caption{Configuration space diagram of $D_n(\Gamma)$ for $\Gamma$ from Fig.\ref{components}.}
\label{DnGdiag}
\end{figure} 
\noindent There are two kinds of connections. Namely, the connections, where the number of particles between the components is the same, and the connections, where one particle moves from $\Gamma_1$ to $\Gamma_2$. Connections of the first kind exist between $D_k(\tilde\Gamma_1)\times D_l(\Gamma_2)$ and $D_k(\Gamma_1)\times D_l(\tilde\Gamma_2)$, where the common part is $D_k(\tilde\Gamma_1)\times D_l(\tilde\Gamma_2)$. Connections of the second kind describe a change in the number of particles, hence they exist between $D_k(\Gamma_1)\times D_l(\tilde\Gamma_2)$ and $D_{k-1}(\tilde\Gamma_1)\times D_{l+1}(\Gamma_2)$, where the common part is $D_{k-1}(\tilde\Gamma_1)\times D_{l}(\tilde\Gamma_2)\times v$. While dealing with distinguishable particles, different distributions of particles give a $n\choose k$-fold splitting of each $\D_k(\Gamma_1)\times \D_l(\tilde\Gamma_2)$. Each component isomorphic to $\D_k(\Gamma_1)\times \D_l(\tilde\Gamma_2)$ is connected with $n-k$ different components that are isomorphic to $\D_{k-1}(\tilde\Gamma_1)\times \D_{l+1}(\Gamma_2)$, see Fig.\ref{D4Gdiag_dist}.
 \begin{figure}[ht]
\includegraphics[width=.9\textwidth]{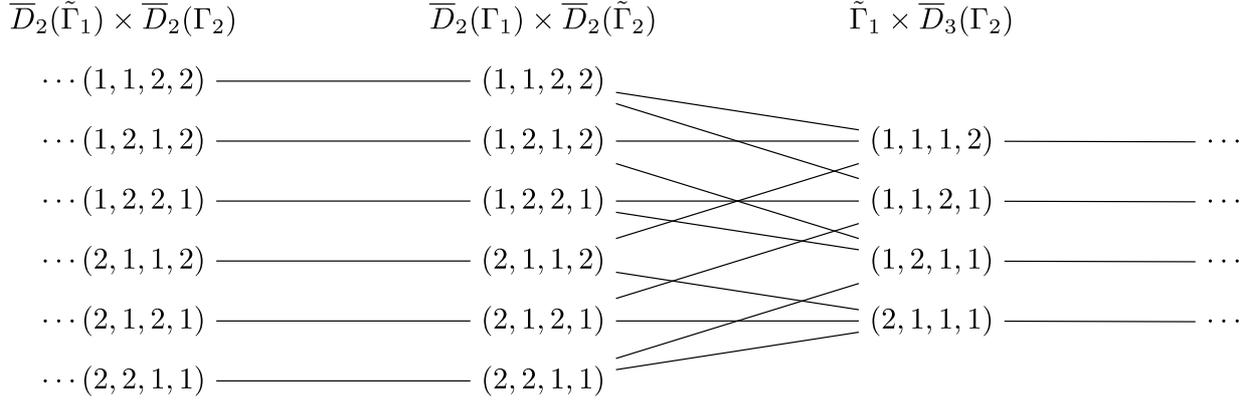}
\caption{A part of the configuration space diagram of $\D_4(\Gamma)$ of a two-component $1$-connected graph. Numbers in brackets denote distribution of particles between the components. For example, $(2,1,1,2)$ in $\D_2(\tilde\Gamma_1)\times \D_2(\Gamma_2)$-column denotes the subcomplex of $\D_4(\Gamma)$, which is isomorphic to $\D_2(\tilde\Gamma_1)\times \D_2(\Gamma_2)$, where particles $1$ and $4$ sit on $\Gamma_2$, and particles $2$ and $3$ sit on $\tilde\Gamma_1$.}
\label{D4Gdiag_dist}
\end{figure} 

As a final remark, note that for a graph, which is a wedge sum of a larger number of components, the presented results for two components can be applied inductively. Choose $\Gamma_1$ to be one of the components and $\Gamma_2'$ to be the wedge sum of the remaining components, $\{\Gamma_i:\ \ i\geq 2\}$. Subgraph $\tilde\Gamma_2'$ is a disjoint sum of graphs. Therefore, $D_k(\Gamma_1)\times D_l(\tilde\Gamma_2')=\sqcup_{i}D_k(\Gamma_1)\times D_l(\tilde\Gamma_i)$ and the configuration space diagram splits in such a node. Detaching inductively the remaining components of $\Gamma_2'$, we obtain the full configuration space diagram.

\section{Second homology group of $D_2(\Gamma)$ and $\D_2(\Gamma)$}\label{sec:two_particles}
In this section, we continue the considerations regarding the example of two particles on a one-connected graphs with two (Fig.\ref{components}), or more components. Using this example, we introduce tools that we finally apply for $D_n(\Gamma)$, where $\Gamma$ is a tree graph. In the end of this section we also give a formula for the second Betti number of $D_2(\Gamma)$, which is a generalisation of the formula by Farber \cite{BF2} for two graphs connected by a single edge.

Consider graph $\Gamma$, which has two components. By the construction of $D_2(\Gamma)$, there are no $3$-cells in the complex, hence $H_2(D_2(\Gamma))$ is free. To compute the second homology, we use Mayer-Vietoris sequence for different components of the configuration space diagram. Let us next briefly introduce the Mayer-Vietoris sequence for the second homology. Let $X$ be any subcomplex of $D_2(\Gamma)$ and let $A$ and $B$ be subcomplexes of $X$ such that $A\cup B=X$. The Mayer-Vietoris sequence for $X$ reads \cite{Hatcher}
\begin{equation}\label{m-v_h2}
0\rightarrow H_2(A\cap B)\xrightarrow{\Phi}H_2(A)\oplus H_2(B)\xrightarrow{\Psi}H_2(X)\xrightarrow{\delta}H_1(A\cap B)\xrightarrow{\Phi}\dots
\end{equation}
Map $\Phi$ acts on $2$-cycles from $A\cap B$ by assigning the same cycle to each summand in the image, i.e. for $x\in\mk{C}_2(A\cap B)$, $\Phi(x)=(x,-x)$. Map $\Psi$ assigns the sum of chains, $\Psi(x,y)=x+y$. The boundary map $\delta$ acts as follows. Every $2$-cycle, $z$, from $A\cup B$ can be written as a sum of $2$-chains form $A$ and $B$ respectively
\[z=x+y,\ x\in\mk{C}_2(A),\ y\in\mk{C}_2(B).\]
Because $\partial z=0$, we have $\partial x=-\partial y$. Chains $\partial x$ and $\partial y$ are $1$-cycles, since $\partial\partial=0$. Moreover, these cycles represent the same element of $H_1(A\cap B)$. In other words, $\delta [z]=[\partial x]=[-\partial y]$. Note that class $[\partial x]$ does not depend on the chosen decomposition of $z$. Because the homology groups over $\ZZ$ are abelian, long exact sequence (\ref{m-v_h2}) can be equivalently written as the short exact sequence
\[0\to\coker(\Phi)\to H_2(X)\to\coker(\Psi)\to0.\]
Recall that cokernel of map $f:\ U\to V$ is defined as $\coker(f)=V/\im(f)$. If the homology groups in the Mayer-Vietoris sequence are free \footnote{In fact, it is enough to require $\coker(\Psi)$ to be free abelian.}, which is the case for $H_2(D_2(\Gamma))$, the sequence splits, i.e.
\[H_2(X)=\coker(\Phi)\oplus \coker(\Psi)=\coker(\Phi)\oplus\im\delta.\]
Finally, we will rather consider elements of $\im\delta$ as elements of $H_2(X)$, i.e. use the isomorphism $\im\delta\cong\coim\delta=H_2(X)/\ker\delta$. Then,
\begin{equation}\label{h2d2_final}
H_2(X)=\coker(\Phi)\oplus\coim\delta.
\end{equation}
\begin{theorem}\label{h2d2}
Let $\Gamma$ be a one-connected graph with two components and let $\Gamma_1,\tilde\Gamma_1,\Gamma_2,\tilde\Gamma_2$ be the components of $\Gamma$, as on Fig.\ref{components}. Then,
\begin{equation}\label{m-v_D2}
\beta_2\left(D_2(\Gamma)\right)=\beta_2(D_2(\Gamma_1))+\beta_2(D_2(\Gamma_2))+\beta_1(\tilde\Gamma_1)\beta_1(\Gamma_2)+\beta_1(\Gamma_1)\beta_1(\tilde\Gamma_2)-\beta_1(\tilde\Gamma_1)\beta_1(\tilde\Gamma_2).
\end{equation}
\end{theorem}
\begin{proof}
Decompose the configuration space part-by-part, as follows
\begin{eqnarray*}
X_0=D_2(\Gamma),\ A_0=D_2(\Gamma_1),\ B_0=(\tilde\Gamma_1\times\Gamma_2)\cup(\Gamma_1\times\tilde\Gamma_2)\cup D_2(\Gamma_2), A_0\cap B_0=\tilde\Gamma_1\times v,\\
X_1=B_0,\ A_1=\tilde\Gamma_1\times\Gamma_2,\ B_1=(\Gamma_1\times\tilde\Gamma_2)\cup D_2(\Gamma_2), A_1\cap B_1=\tilde\Gamma_1\times\tilde\Gamma_2,\\
X_2=B_1,\ A_2=\Gamma_1\times\tilde\Gamma_2,\ B_2=D_2(\Gamma_2),\ A_2\cap B_2=\tilde\Gamma_2\times v.
\end{eqnarray*}
The ansatz is to write Mayer-Vietoris sequence for each $X_i=A_i\cup B_i$ and proceed inductively, beginning with $X_2$. Namely, 
\[\coker\Phi_2=(H_2(A_2)\oplus H_2(B_2))/\im\Phi_2=(H_1(\Gamma_1)\otimes H_1(\tilde\Gamma_2))\oplus H_2(D_2(\Gamma_2)),\]
where in $H_2(A_2)$ we used the K\"unneth theorem. The image of $\Phi_2$ is trivial, because $H_2(A_2\cap B_2)=H_2(\tilde\Gamma_2\times v)=0$, therefore $\coim\Phi_2=0$. Next, we give a characterisation of elements of $\coim\delta$ for $i=2$. Recall that
\[\delta_2:\ H_2(X_2)\to H_1(A_2\cap B_2)\cong H_1(\tilde\Gamma_2).\]
Denote by $z$ a representative of $H_2(X_2)$. Two-cycle $z$ can be decomposed as a sum of $2$-chains from $A_2$ and $B_2$ respectively
\begin{equation}\label{decomp2}
z=x+y,\ x\in\mk{C}_2(A_2),\ y\in\mk{C}_2(B_2).
\end{equation}
The above decomposition is in this case unique, because there are no $2$-cells in the subcomplex $A_2\cap B_2$. The boundary map assigns to $z$ the $1$-cycle $\delta_2 y$. Let us keep for the moment $\coim\delta_2$. Then
\[H_2(B_2)=H_2(X_2)=H_2(D_2(\Gamma_2))\oplus\coim\delta_2\oplus(H_1(\Gamma_1)\otimes H_1(\tilde\Gamma_2))\]
Let us proceed with $i=1$. We have $\coim\Phi_1=H_1(\tilde\Gamma_1)\otimes H_1(\tilde\Gamma_2)$, hence
\[\coker\Phi_1=\left(H_2(B_2)\oplus (H_1(\tilde\Gamma_1)\oplus H_1(\Gamma_2)\right)/\left(H_1(\tilde\Gamma_1)\otimes H_1(\tilde\Gamma_2)\right).\]
Note that the quotient does not affect $H_2(D_2(\Gamma_2))$ and $\coim\delta_2$. This is because every $2$-cycle from $\coim\Phi_1$ is of the form 
\begin{equation}\label{coimphi1}
z=c\otimes c',\ [c]\in H_1(\tilde\Gamma_1),\  [c']\in H_1(\tilde\Gamma_2).
\end{equation}
Cycle $z$ is a $2$-cycle from $D_2(\tilde\Gamma_1\times \tilde\Gamma_2)$. Therefore, decomposition (\ref{decomp2}) for such a $2$-cycle yields $y=0$, i.e. $z\in\ker\delta_2$. Moreover, none of the representatives of a homology class from $H_2(D_2(\Gamma_2))$ is of the form (\ref{coimphi1}). Therefore,
\[H_2(B_1)=H_2(D_2(\Gamma_2))\oplus\coim\delta_2\oplus\coim\delta_1\oplus\frac{(H_1(\Gamma_1)\otimes H_1(\tilde\Gamma_2))\oplus (H_1(\tilde\Gamma_1)\otimes H_1(\Gamma_2))}{H_1(\tilde\Gamma_1)\otimes H_1(\tilde\Gamma_2)}.\]
Finally, for $i=0$, by an analogical reasoning as in the case of $i=2$, we have $\im\Phi_0=0$, hence
\[H_2(D_2(\Gamma))=H_2(D_2(\Gamma_1))\oplus H_2(D_2(\Gamma_2))\oplus\coim\delta\oplus\frac{(H_1(\Gamma_1)\otimes H_1(\tilde\Gamma_2))\oplus (H_1(\tilde\Gamma_1)\otimes H_1(\Gamma_2))}{H_1(\tilde\Gamma_1)\otimes H_1(\tilde\Gamma_2)},\]
where $\coim\delta=\oplus_{i=0}^2\coim\delta_i$. Using a theorem by Farber \cite{BF1}, we prove in lemma \ref{lemma_partial_d2} that $\coim\delta=0$, which completes the proof.
\end{proof}
\noindent The last thing to show is the fact that $\im\delta=0$. To this end, we consider a specific over-complete basis of the second homology group.  First, we briefly review the known facts about the second homology group of the two-particle configuration spaces.
\begin{theorem}[Farber \cite{BF1}]\label{farber}
For a planar graph $\Gamma$, there exists a basis of $H_2(\D_2(\Gamma))$, where the representatives are of the form 
\[z=c\otimes c', [c],[c']\in H_1(\Gamma).\]
 Moreover, cycles $c$ and $c'$ are necessarily disjoint, i.e. for
\[c=\sum_ia_i e_i,\ c'=\sum_j b_j e_j',\]
we have $e_i\cap e_j'=\emptyset$ for all $i,j$. Then, $z=\sum_{i,j}a_i b_j e_i\times e_j'$. Analogical result holds for a basis of $H_2(D_2(\Gamma))$.
\end{theorem}
\noindent Hence, for planar graphs the over-complete basis that we will consider, consists of all possible pairs of disjoint cycles in $\Gamma$. The construction of an over-complete basis for two particles on non-planar graphs requires a slight refinement. Recall first a theorem by Kuratowski \cite{kuratowski}, which states that every non-planar graph contains a subgraph that is isomorphic to graph $K_{3,3}$ or $K_5$. Furthermore, it was first shown by Abrams \cite{AbramsPhD} that graphs $K_{3,3}$ and $K_5$ are the only possible graphs, whose two-particle distinguished (discrete) configuration spaces are closed surfaces. Hence, in the case of distinguishable particles, for every subgraph of $\Gamma$, which is isomorphic to $K_{3,3}$ or $K_5$, we add an element to the over-complete basis, which is isomorphic to $\D_2(K_{3,3})$ or $\D_2(K_{5})$ respectively. For indistinguishable particles, there are no graphs, whose two-particle configuration spaces are isomorphic to a closed surface, hence it is enough to consider only the products of cycles as an over-complete basis of $H_2(D_2(\Gamma))$.
\begin{lemma}\label{lemma_partial_d2}
Let $\Gamma$ be a one-connected graph with two components and $\{(A_i,B_i)\}_{i=0}^2$ be the subcomplexes of $D_2(\Gamma)$ from the proof of theorem \ref{h2d2}. Moreover, let $\delta_i$ be the boundary map from the Mayer-Vietoris sequence
\[\delta_i:\ H_2(A_i\cup B_i)\to H_1(A_i\cap B_i).\] 
Then, $\im\delta_i=0$ for all $i$.
\end{lemma}
\begin{proof}
Because $H_2(D_2(\Gamma))$ contains $\coim\delta$ as an independent contribution, theorem \ref{farber} applies also to $2$-cycles representing $\coim\delta$. Let us begin with  The strategy for the proof is to show that all possible products of disjoint $1$-cycles in $X_i=A_i\cup B_i$ are in $\ker\delta_i$.
 
For $\delta_2:\ H_2((\tilde\Gamma_1\times\Gamma_2)\cup D_2(\Gamma_2))\to H_1(\tilde\Gamma_2\times v)$, a $2$-cycle, which does not belong to $\ker\delta_2$, has a nonzero part in both $\mk{C}_2(\tilde\Gamma_1\times\Gamma_2)$ and $\mk{C}_2(D_2(\Gamma_2))$. Let $z$ be a representative of such a $2$-cycle from the basis in theorem \ref{farber}, i.e. $z=c\otimes c'$. Note that every $1$-cycle in $\Gamma$ can be written as a sum of cycles that are wholly contained in $\Gamma_1$ or $\Gamma_2$. Therefore, the only possibility for the choice of $z$ to contain cells from both $A_2$ and $B_2$ is to take $c\in\mk{C}_1(\Gamma_2)$ that contains an edge adjacent to $v$. However, $c$ and $c'$ are disjoint, hence $c'$ must be contained in $\tilde\Gamma_1$ or $\tilde\Gamma_2$. This is a contradiction, because then $z\in\mk{C}_2(\tilde\Gamma_1\times\Gamma_2)$ or $z\in\mk{C}_2(D_2(\Gamma_2))$ respectively, which means that $z\in\ker\delta_2$. Analogical reasoning for $A_0$ and $B_0$ leads to the conclusion that $\im\delta_0=0$. Finally, consider $\delta_1:\ H_2((\Gamma_1\times\tilde\Gamma_2)\cup (\tilde\Gamma_1\times\Gamma_2))\to H_1(\tilde\Gamma_1\times\tilde\Gamma_2)$. The desired $2$-cycle can be a product of 
\[c\in\mk{C}_1(\Gamma_1),\ c\cap v\neq\emptyset,\ c'\in\mk{C}_1(\tilde\Gamma_2),\ {\rm or}\ c\in\mk{C}_1(\tilde\Gamma_1),\ c'\in\mk{C}_1(\Gamma_2),\ c'\cap v\neq\emptyset.\]
It is straightforward to see that in both cases $z\in\mk{C}_2(\Gamma_1\times\tilde\Gamma_2)$ or $z\in\mk{C}_2(\tilde\Gamma_1\times\Gamma_2)$ respectively. Therefore, $z\in\ker\delta_1$.
\end{proof}
\noindent The formula for the second Betti number for two distinguishable particles can be obtained using the same strategy as in the proof of theorem \ref{h2d2}.
\begin{coro}
Let $\Gamma$ be a one-connected graph with two components and let $\Gamma_1,\tilde\Gamma_1,\Gamma_2,\tilde\Gamma_2$ be the components of $\Gamma$, as on Fig.\ref{components}. Then,
\begin{equation}\label{m-v_D2dist}
\beta_2\left(\D_2(\Gamma)\right)=\beta_2(\D_2(\Gamma_1))+\beta_2(\D_2(\Gamma_2))+2\left(\beta_1(\tilde\Gamma_1)\beta_1(\Gamma_2)+\beta_1(\Gamma_1)\beta_1(\tilde\Gamma_2)-\beta_1(\tilde\Gamma_1)\beta_1(\tilde\Gamma_2)\right).
\end{equation}
\end{coro}
\begin{proof}
Denote by $(\X_i,\A_i,\B_i),\ i=0,1,2$ the following subcomplexes of $\D_2(\Gamma)$.
\begin{eqnarray*}
\X_0=\D_2(\Gamma),\ \A_0=\D_2(\Gamma_1),\  \A_0\cap \B_0=(\tilde\Gamma_1\times v)\sqcup(v\times\tilde\Gamma_1),\\
\X_1=\B_0,\ \A_1=(\tilde\Gamma_1\times\Gamma_2)\sqcup(\Gamma_2\times\tilde\Gamma_1),\ \A_1\cap \B_1=(\tilde\Gamma_1\times\tilde\Gamma_2)\sqcup(\tilde\Gamma_2\times\tilde\Gamma_1),\\
\X_2=\B_1,\ \A_2=(\Gamma_1\times\tilde\Gamma_2)\sqcup(\tilde\Gamma_2\times\Gamma_1),\ \B_2=\D_2(\Gamma_2),\ \A_2\cap \B_2=(\tilde\Gamma_2\times v)\sqcup(v\times\tilde\Gamma_2).
\end{eqnarray*}
The contribution to $H_2(\D_2(\Gamma))$ from products of $1$-cycles from $\Gamma_1$ and $\Gamma_2$ is twice the corresponding contribution from $H_2(D_2(\Gamma))$. This is because $\A_1$ and $\A_2$ are disjoint sums of subcomplexes that are both isomorphic to $A_1$ and $A_2$ from the proof of theorem \ref{h2d2} respectively, and the relations imposed by $\im\Phi_i$ do not mix the components. In other words, 
\[\A_i\cong A_i\sqcup A_i,\ \A_i\cap \B_i\cong (A_i\cap B_i)\sqcup (A_i\cap B_i),\ H_k(\A_i)\cong H_k(A_i)\oplus H_k(A_i),\ k=1,2.\]
The argument showing that $\coim\delta=0$ passes without changes for the part of the over-complete basis that consists of products of disjoint $1$-cycles. For completeness, consider the additional elements of the over-complete basis for $\D_2(\Gamma)$ stemming from all subgraphs of $\Gamma$ that are isomorphic to $K_{3,3}$ or $K_5$. Note that such a subgraph must be necessarily contained either in $\Gamma_1$ or in $\Gamma_2$. Hence, the $2$-cycles that are isomorphic to $\D_2(K_{3,3})$ or $\D_2(K_{5})$ are always contained $\D_2(\Gamma_1)$ or $\D_2(\Gamma_2)$ respectively, i.e. mapped by $\delta$ to zero. As we explain in section \ref{sec:partial}, $\im\delta$ is no longer trivial for the three-particle case.
\end{proof}
Note that the derived formulae for $\beta_2\left(D_2(\Gamma)\right)$ and $\beta_2\left(\D_2(\Gamma)\right)$ can be written in a simpler form using $\mu_i:=\beta_1(\Gamma_i)-\beta_1(\tilde\Gamma_i)$. Then, 
\begin{equation}\label{h2d2_indist}
\beta_2\left(D_2(\Gamma)\right)=\beta_2(D_2(\Gamma_1))+\beta_2(D_2(\Gamma_2))+\beta_1(\Gamma_1)\beta_1(\Gamma_2)-\mu_1\mu_2.
\end{equation}
\begin{equation}\label{h2d2_dist}
\beta_2\left(\D_2(\Gamma)\right)=\beta_2(\D_2(\Gamma_1))+\beta_2(\D_2(\Gamma_2))+2\left(\beta_1(\Gamma_1)\beta_1(\Gamma_2)-\mu_1\mu_2\right).
\end{equation}
Numbers $\mu_i$ are the numbers of cycles lost after detaching vertex $v$. If $\Gamma_i$ is connected with $v$ with $E_i$ edges, then $\mu_i=E_i-1$.

The above formulae can be easily extended to a graph with more than two components. Assume that $\Gamma$ has three components, $\{\Gamma_i\}_{i=1}^3$. Then, $\Gamma$ can be viewed as a two-component graph, where the first component is $\Gamma_1$ and the second component is the wedge sum of $\Gamma_2$ and $\Gamma_3$, which we denote by $\Gamma_{23}$. Moreover, $\tilde\Gamma_{23}=\Gamma_2\sqcup\Gamma_3$. Next, we apply formula (\ref{m-v_D2}), using the fact that $\beta_1(\tilde\Gamma_{23})=\beta_1(\tilde\Gamma_2)+\beta_1(\tilde\Gamma_3)$ and $\mu_{23}=\mu_2+\mu_3$.
\begin{eqnarray*}
\beta_2\left(D_2(\Gamma)\right)=\beta_2(D_2(\Gamma_1))+\beta_2(D_2(\Gamma_{23}))+\beta_1(\Gamma_1)\left(\beta_1(\Gamma_2)+\beta_1(\Gamma_3)\right)-\mu_1(\mu_2+\mu_3).
\end{eqnarray*}
Finally, we put $\beta_2(D_2(\Gamma_{23}))=\beta_2(D_2(\Gamma_2))+\beta_2(D_2(\Gamma_3))+\beta_1(\Gamma_2)\beta_1(\Gamma_3)-\mu_2\mu_3$. Then, 
\begin{eqnarray*}
\beta_2\left(D_2(\Gamma)\right)=\sum_i\beta_2(D_2(\Gamma_i))+\sum_{i< j}(\beta_1(\Gamma_i)\beta_1(\Gamma_j)-\mu_i\mu_j).
\end{eqnarray*}
For distinguishable particles, we have
\begin{eqnarray*}
\beta_2\left(\D_2(\Gamma)\right)=\sum_i\beta_2(\D_2(\Gamma_i))+2\sum_{i< j}(\beta_1(\Gamma_i)\beta_1(\Gamma_j)-\mu_i\mu_j).
\end{eqnarray*}

\section{Tree graphs}\label{sec:tree}
Let us first compute the second homology group for $n$ particles on a tree, which consists of two star graphs, $S$ and $S'$, connected by an edge. We will assume that both $S$ and $S'$ are sufficiently subdivided for $n$ particles, see Fig.\ref{tree-components}.
 \begin{figure}[ht]
 
\includegraphics[width=0.6\textwidth]{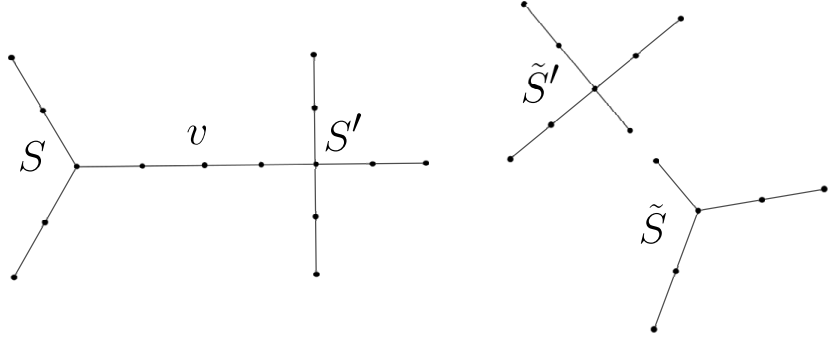}
\caption{A tree graph, which is a wedge of two star graphs, and its components for $n=3$. Compare with Fig.\ref{components}.}
\label{tree-components}
\end{figure} 
The general procedure for computing $H_2(D_n(T))$ will be to decompose the $n$-particle configuration space part-by-part in the following way. 
\begin{eqnarray*}
X_0=D_n(T),\ A_0=D_n(S),\ B_0&=&(D_{n-1}(\tilde S)\times S')\cup(D_{n-1}(S)\times\tilde S')\cup\dots\cup D_n(S'), \\
X_1'=B_0,\ A_1'=D_{n-1}(\tilde S)\times S',\ B_1'&=&(D_{n-1}(S)\times\tilde S')\cup\dots\cup D_n(S'),\\
X_1=B_1',\ A_1=D_{n-1}(S)\times \tilde S',\ B_1&=&(D_{n-2}(\tilde S)\times S')\cup\dots\cup D_n(S'),\\
&\vdots& \\
X_{n-1}=B'_{n-1},\ A_{n-1}=S\times D_{n-1}(\tilde S'),\ B_{n-1}&=&D_n(S').
\end{eqnarray*}
We distinguish two kinds of subcomplexes with respect to the type of their intersections. A pair $(A_k',B_k')$ of the first kind, describes subcomplexes that have the same number of particles on $S$ and $S'$, i.e.
\begin{eqnarray*} 
A_k'=D_{n-k}(\tilde S)\times D_k(S'),\ B_k'=(D_{n-k}(S)\times D_k(\tilde S'))\cup\dots\cup D_n(S'),\\ A_k'\cap B_k'=D_{n-k}(\tilde S)\times D_k(\tilde S'),\ X_k'=A_k'\cup B_k',\ k\in\{1,2,\dots,n-1\}.
\end{eqnarray*}
The second kind of subcomplexes describes pairs, where the numbers of particles on $S$ and $S'$ are different. 
\begin{eqnarray*} 
A_k=D_{n-k}(S)\times D_k(\tilde S'),\ B_k=(D_{n-k-1}(\tilde S)\times D_{k+1}(S'))\cup\dots\cup D_n(S'),\\ A_k\cap B_k=D_{n-k-1}(\tilde S)\times D_k(\tilde S')\times v,\ X_k=A_k\cup B_k,\ k\in\{0,1,\dots,n-1\}.
\end{eqnarray*}
In the above notation, $X_0\supset X_1'\supset X_1\supset\dots\supset X_n$. Next, we write the Mayer-Vietoris sequence for each pair of subcomplexes, as in (\ref{m-v_h2}).
\begin{equation}\label{m-v_h2tree}
0\rightarrow H_2(A_k\cap B_k)\xrightarrow{\Phi_k}H_2(A_k)\oplus H_2(B_k)\xrightarrow{\Psi_k}H_2(X_k)\xrightarrow{\delta_k}H_1(A_k\cap B_k)\xrightarrow{}\dots
\end{equation}
We use the fact that $H_3(X_k)=0$, since there are no $3$-cells in the Morse complex. Because homology groups for tree graphs are free, the Mayer-Vietoris sequence splits and we have
\begin{equation}\label{h2xk_tree}
H_2(X_k)=\coker(\Phi_k)\oplus\coim\delta_k,\ H_2(X'_k)=\coker(\Phi'_k)\oplus\coim\delta'_k.
\end{equation}
There are a few differences between the case of tree graphs and a general one-connected graph with two components, which allow to compute the second homology for any number of particles. The first simplification comes from the fact that $D_k(S)$ is homotopy equivalent to $D_k(\tilde S)$ for $k<n$, and $D_k(S)$ is homotopy equivalent to a wedge of circles. The same holds for distinguishable particles. For $n=2$, formulae (\ref{h2d2_dist}) and (\ref{h2d2_indist}) yield $\beta_2(D_2(T))=\beta_2(\D_2(T))=0$. For $n=3$, note that all subcomplexes in the configuration space diagram have trivial homology groups in the corresponding Mayer-Vietoris sequences. Therefore, $\beta_2(D_3(T))=\beta_2(\D_3(T))=0$ and the first nontrivial case is $n=4$.

\begin{theorem}\label{th:h2_tree2comp}
Let $T$ be a tree graph with two components $S$ and $S'$ (Fig.\ref{tree-components}). The rank of the second homology group for $n$ indistinguishable particles on $T$ is 
\begin{equation}\label{h2_tree2comp}
\beta_2^{(n)}(S,S')=\sum_{l=2}^{n-2}\left(\beta_1^{(l)}(S)-\beta_1^{(l-1)}(S)\right)\beta_1^{(n-l)}(S'),
\end{equation}
where $\beta_1^{(k)}(S)$ is the rank of the first homology group for $k$ particles on star graph $S$, given in equation (\ref{h1star_ind}). 
\end{theorem}
\begin{proof}
Consider two Mayer-Vietoris sequences for two consecutive subcomplexes, $X_k$ and $X_k'$. We will obtain a recurrence relation for $H_2(B_k)$. Maps from the sequence for $X_k$ are
\begin{eqnarray*}
\Phi_k:\ H_2(D_{n-k-1}(\tilde S)\times D_k(\tilde S')\times v)\to H_2(D_{n-k}(S)\times D_k(\tilde S'))\oplus H_2(B_k), \\
\delta_k:\ H_2(X_k)\to H_1(D_{n-k-1}(\tilde S)\times D_k(\tilde S')\times v).
\end{eqnarray*}
Because $H_3(X_k)=0$, the Mayer-Vietoris sequence implies that map $\Phi_k$ is injective. Hence, $\im\Phi_k\cong H_2(A_k\cap B_k)$. Moreover, $\im\delta_k=0$. This is because each element of the over-complete basis of $2$-cycles is a chain, which is properly contained in $D_{n-k}(\tilde S)\times D_k(S')$ or $D_{n-k}(S)\times D_k(\tilde S')$ for some $k$. Therefore, we have
\[H_2(B_k')=\coker\Phi_k\cong \left(\left(H_1(D_{n-k}(S))\otimes H_1(D_k(\tilde S'))\right)\oplus H_2(B_k)\right)/\im\Phi_k.\]
The quotient can be realised as follows. Any element of $\coim\Phi_k$ can be written as a tensor product of chains of the following form
\[[c\otimes c']\times v:\ [c]\in H_1(D_{n-k-1}(\tilde S)),\ [c']\in H_1(D_k(\tilde S')).\]
Furthermore, each such $2$-cycle can be written as $2$-cycle $(c\times v)\otimes c'$, which belongs to $\mk{C}_2(A_k)$, or $2$-cycle $c\otimes (c'\times v)$, which belongs to $\mk{C}_2(B_k)$. Map $\Phi_k$ acts on the homology classes as
\[\Phi_k([c\otimes c']\times v)=([(c\times v)\otimes c'],[-c\otimes (c'\times v)]).\]
On the other hand, every element of $H_2(A_k)$ can be decomposed in the basis of the tensor product
\[[\tilde c\otimes \tilde c']:\ [\tilde c]\in H_1(D_{n-k}(S)),\ [\tilde c']\in H_1(D_k(\tilde S')).\]
By the injectivity of $\Phi_k$, cycles $(c_1\times v)\otimes c_1'$ and $(c_2\times v)\otimes c_2'$ represent different classes in $H_2(A_k)$ if $[c_1]\neq[c_2]$ or $[c_1']\neq[c_2']$. Therefore, from every element $[a]\in H_2(A_k)$ we can extract in a unique way the part, which belongs to $H_2(A_k\cap B_k)$, i.e.
\begin{equation}\label{a_decomp}
[a]=\sum_{[c],[c']} [(c\times v)\otimes c']+[\tilde a].
\end{equation}
Therefore, for a pair $([a],[b])\in H_2(A_k)\oplus H_2(B_k)$, where $[a]$ is decomposed, as in (\ref{a_decomp}), we have 
\[\left([a],[b]\right)\sim\Big([\tilde a],[b]+\sum_{[c],[c']} [(c\times v)\otimes c']\Big)\]
under the quotient by $\im\Phi_k$. Moreover, pairs, where $[a]=[\tilde a]$, yield different equivalence classes for different $a,b$. This means that the quotient by $\im\Phi_k$ can be realised by taking the quotient by $\coim\Phi_k$ only on $H_2(A_k)$. In other words,
\[H_2(B_{k'})\cong \frac{H_1(D_{n-k}(S))\otimes H_1(D_k(\tilde S'))}{H_1(D_{n-k-1}(\tilde S))\otimes H_1(D_k(\tilde S'))}\oplus H_2(B_k).\]
A similar result holds for pair $(A_k',B_k')$, i.e. $\Phi_k'$ is injective and $\im\delta_k'=0$. The quotient by $\Phi_k'$ reads
\[H_2(B_{k-1})\cong \frac{H_1(D_{n-k}(\tilde S))\otimes H_1(D_k(S'))}{H_1(D_{n-k}(\tilde S))\otimes H_1(D_k(\tilde S'))}\oplus H_2(B_k')\cong H_2(B_k').\]
Subtracting ranks in the equation for $H_2(B_k')$, we obtain a recurrence relation for ranks of $H_2(B_k)$ and $H_2(B_{k-1})$
\begin{equation}\label{recurrence_2star}
\beta_2(B_{k-1})=\left(\beta_1^{(k)}(S)-\beta_1^{(k-1)}(S)\right)\beta_1^{(n-k)}(S')+\beta_2(B_k).
\end{equation}
The initial condition is $\beta_2(B_{n-1})=0$.
\end{proof}
\noindent The following theorem shows how to compute the second homology group of any tree.
\begin{theorem}\label{thm:h2_tree}
Let $T$ be a tree graph and let a pair $(S,S')$ denote the subgraph of $T$, which consists of two star graphs $S$ and $S'$ and the unique path in $T$ that connects the essential vertices of $S$ and $S'$ (see Fig.\ref{fig:h2_tree}a). Then,
\[H_2(D_n(T))\cong\bigoplus_{(S,S')\subset T}H_2(D_n((S,S'))).\]
\end{theorem}
\begin{proof}
The strategy for the proof is to show that every cycle from the over-complete basis of $H_2(D_n(T))$ is homologically equivalent to a $2$-cycle from $D_n(S,S')$ for a pair of star subgraphs of $T$. Assume first that every star subgraph of $T$ is sufficiently subdivided for $n$ particles. This in particular means that the edges that connect essential vertices are subdivided twice as much as the condition of sufficient subdivision requires. Every $2$-cycle from the over-complete basis of $H_2(D_n(T))$ is isomorphic to a tensor product of two chains, each describing exchange of particles a $Y$-subgraph of $T$, and the remaining $n-4$ particles distributed on free vertices of $T$. Let $S$ and $S'$ be the star-subgraphs of $T$ that contain the two $Y$-subgraphs, where the particles exchange (Fig.\ref{fig:h2_tree}a).
 \begin{figure}[ht]
  
\includegraphics[width=0.8\textwidth]{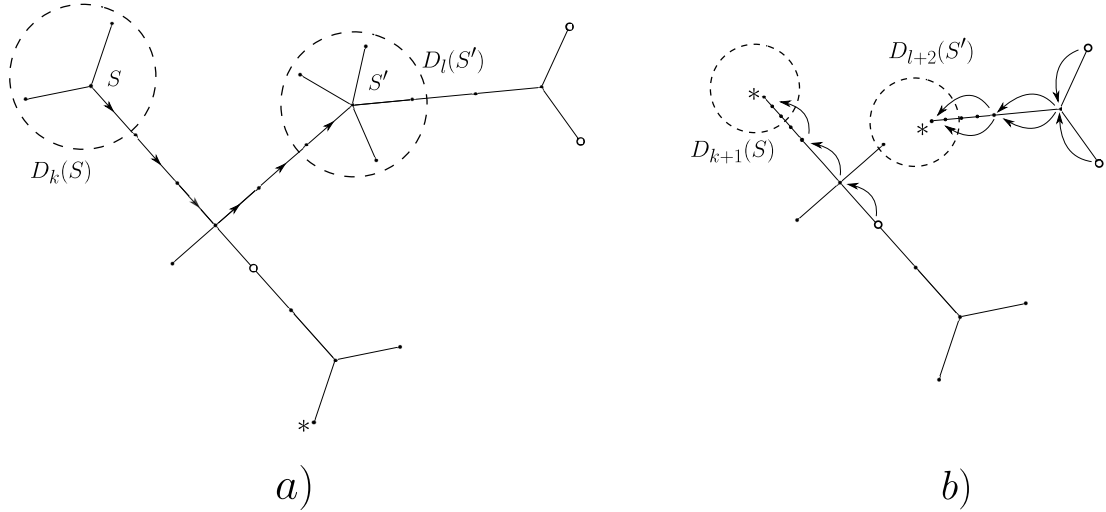}
\caption{Illustration for the proof of Theorem \ref{thm:h2_tree}. White vertices denote vertices that are occupied by particles from the outside of $D_k(S)$ and $D_l(S')$. Figure b) shows the construction of the path that brings each particle to a configuration space of one of the star subgraphs, where the particles exchange.}
\label{fig:h2_tree}
\end{figure}
The remaining particles are distributed on the remaining vertices of $T$. Some of them may occupy free vertices of $S$ or $S'$. Assume that $k-2$ out of free particles occupy star graph $S$ and $l-2$ free particles occupy $S'$. The remaining $n-(k+l)$ particles are distributed outside $S$ and $S'$. The element of the over-complete basis of $H_2(D_n(T))$ that corresponds to such a situation is of the form
\[\sigma=(c\otimes c')\times\{v_1,\dots,v_{n-k-l}\},\ [c]\in H_1(D_k(S)),\ [c']\in H_1(D_l(S')),\ \{v_1,\dots,v_{n-k-l}\}\notin S,S'.\] 
We will next give a construction of a path in $D_{n}(T)$ that connects point $\{v_1,\dots,v_{n-k-l}\}$ with a point, where all the particles are distributed on star graphs $S$ and $S'$. To this end, remove from $T$ subgraphs $S$ and $S'$ by removing star subgraphs that are sufficiently subdivided for $k$ and $l$ particles respectively. After removing the star subgraphs, graph $T$ decomposes into a number of connected components (see Fig.\ref{fig:h2_tree}b). Each component has a number of vertices of valence one, where the star graphs were attached. Order these vertices according to their distance from the root in $T$. For each component, choose the new root to be the vertex, which was the closest one to the original root in $T$ (Fig.\ref{fig:h2_tree}b). Finally, move all particles in the components to the roots. The resulting configuration is a configuration, where all particles are distributed on star graphs $S$ and $S'$.
\end{proof}
\noindent As a consequence, the rank of the second homology group for $n$ particles reads
\begin{equation}\label{h2_tree}
\beta_2^{(n)}(T)=\sum_{S,S'\subset T}\beta_2^{(n)}(S,S').
\end{equation}
An analogical result holds for all the higher homology groups. For the $m$th homology one has to take the sum over all tree subgraphs of $T$ that contain exactly $m$ star subgraphs, i.e.
\[H_m(D_n(T))\cong\bigoplus_{T'\subset T:\ \#T'=m}H_m(D_n(T')).\]
The proof is the same as the proof of theorem \ref{thm:h2_tree} -- for every $m$-cycle from the over-complete basis decompose $T$ by removing the star graphs that belong to $T'$ and move the remaining particles within the components.

Hence, the problem of computing $H_m(D_n(T))$ for any tree boils down to the problem of computing the $m$th homology for a tree containing $m$ essential vertices. To this end, we consider a bipartition $(S,T')$ of such a tree, where $S$ is one of the star graphs of valence $1$ in the sense of the scheme of connections in the tree (see Fig.\ref{tree}b), and $T'$ is the tree with graph $S$ removed.
\begin{theorem}\label{th:hk_treekcomp}
Let $T'$ be a tree graph with $m-1$ essential vertices. Construct a tree graph $T$ with $m$ essential vertices as a wedge sum of $T'$ and a star graph $S$, i.e. $T=(T'\sqcup S)/\sim$, where gluing map $\sim$ identifies two vertices of valence $1$ in $T'$ and $S$. The rank of the $m$th homology group for $n$ indistinguishable particles on $T$ is 
\begin{equation}\label{hk_treekcomp}
\beta_m^{(n)}(S,T')=\sum_{l=2}^{n-2}\left(\beta_1^{(l)}(S)-\beta_1^{(l-1)}(S)\right)\beta_{m-1}^{(n-l)}(T').
\end{equation}
\end{theorem}
\begin{proof}
Note first that for such a tree graph, we have $H_{m+1}(D_n(T))=0$, because the dimension of the corresponding Morse complex is $m$. Consider $T$ as a $1$-connected graph with two components, where the components are $S$ and $T'$. Vertex $v$ connecting the components has valence $2$. Next, construct the sequence of subcomplexes 
\[D_n(T)=X_0\supset X_1'\supset X_1\supset\dots\supset X_{n-1}=\left(S\times D_{n-1}(\tilde T')\right)\cup D_n(T'),\]
as in the case of two star graphs. The Mayer-Vietoris sequence for each subcomplex reads
\begin{equation*}
0\rightarrow H_m(A_k\cap B_k)\xrightarrow{\Phi_k}H_m(A_k)\oplus H_m(B_k)\xrightarrow{\Psi_k}H_m(X_k)\xrightarrow{\delta_k}H_{m-1}(A_k\cap B_k)\xrightarrow{}\dots
\end{equation*}
The above sequence splits and we have
\begin{equation*}
H_m(X_k)=\coker(\Phi_k)\oplus\coim\delta_k,\ H_m(X'_k)=\coker(\Phi'_k)\oplus\coim\delta'_k.
\end{equation*}
Consider two Mayer-Vietoris sequences for two consecutive subcomplexes, $X_k$ and $X_k'$. We will obtain a recurrence relation for $H_m(B_k)$, as in the proof of theorem \ref{th:h2_tree2comp}. Maps from the sequence for $X_k$ are
\begin{eqnarray*}
\Phi_k:\ H_m(D_{n-k-1}(\tilde S)\times D_k(\tilde T')\times v)\to H_m(D_{n-k}(S)\times D_k(\tilde T'))\oplus H_m(B_k), \\
\delta_k:\ H_m(X_k)\to H_{m-1}(D_{n-k-1}(\tilde S)\times D_k(\tilde T')\times v).
\end{eqnarray*}
The corresponding maps for $X_k'$ read
\begin{eqnarray*}
\Phi_k':\ H_m(D_{n-k}(\tilde S)\times D_k(\tilde T'))\to H_m(D_{n-k}(\tilde S)\times D_k(T'))\oplus H_m(B'_k), \\
\delta_k':\ H_m(X'_k)\to H_{m-1}(D_{n-k}(\tilde S)\times D_k(\tilde T')).
\end{eqnarray*}
Again, from the construction of the over-complete basis, every $m$-cycle from the basis is contained in $A_k$ or $B_k$, hence $\im\delta_k=0$ and $\im\delta_k'=0$ for all $k$. Hence, the homology groups of the subcomplexes are
\begin{eqnarray*}
H_m(B_k')\cong \left((H_1(D_{n-k}(S))\otimes H_{m-1}(D_k(\tilde T')))\oplus H_m(B_k)\right)/\im\Phi_k, \\
H_m(B_{k-1})\cong \left((H_1(D_{n-k}(\tilde S))\otimes H_{m-1}(D_k(T')))\oplus H_m(B_k')\right)/\im\Phi'_k.
\end{eqnarray*}
As in the proof of theorem \ref{th:h2_tree2comp}, the above quotients can be realised by taking the quotient by $\coim\Phi_k$ and $\coim\Phi_k'$ on $H_m(A_k)$ and $H_m(A_k')$ respectively. By doing so, we get $H_m(B_{k-1})\cong H_m(B_k')$ and 
\[H_m(B_k')\cong H_m(B_k')\oplus\left(H_1(D_{n-k}(S))\otimes H_{m-1}(D_k(\tilde T'))\right)/\left(H_1(D_{n-k-1}(\tilde S))\otimes H_{m-1}(D_k(\tilde T'))\right).\]
This gives us the following recursive equation for $\beta_m(B_k)$
\[\beta_m(B_{k-1})=\beta_m(B_k)+\left(\beta_1^{(n-k)}(S)-\beta_1^{(n-k-1)}(S)\right)\beta_{m-1}^{(k)}(T')\]
with the initial condition $\beta_m(B_{n-1})=0$. The solution is equation $(\ref{hk_treekcomp})$.
\end{proof}
\noindent Note that equation $(\ref{hk_treekcomp})$ allows one to express $\beta_m^{(n)}(T)$ for a tree with $m$ essential vertices by the ranks of the first homology groups for different numbers of particles on the star subgraphs contained in $T$. To this end, one has to apply equation $(\ref{hk_treekcomp})$ recursively, until all star subgraphs of $T$ are removed. One can check by a straightforward calculation that the solution to such a recursion with the initial condition $\beta_1^{(n)}(T)=\beta_1^{(n)}(S)$ is 
\begin{equation}\label{eq:hm_tree}
\beta_m^{(n)}(T)=\sum_{i=0}^{m-1}(-1)^i {{m-1}\choose i}\sum_{l_1+\dots+l_m=n-i,l_j\geq2}\beta_1^{(l_1)}(S^{(1)})\beta_1^{(l_2)}(S^{(2)})\dots\beta_1^{(l_m)}(S^{(m)}),
\end{equation}
where $\{S^{(j)}\}_{j=1}^m$ is the set of all star subgraphs from $T$.

\section{Two graphs connected by a single edge}\label{sec:one_edge}
In this section we continue the line of thought from the previous sections and show how to compute the second homology group for $n$ particles on a graph, which consists of two arbitrary graphs, that are connected by a single edge. In other words, vertex $v$ on figure \ref{components} has valence $1$. The results of this section can be viewed as another generalisation of the formula for the second homology group for two particles on such a graph from \cite{BF2}. As in previous sections, we consider the the Mayer-Vietoris sequence for pairs $(A_k,B_k)$, $(A_k',B_k')$ from the configuration space diagram, where 
\begin{equation}\label{eq:single_edge_decomp}
D_n(\Gamma_1)=X_0\supset X_1'\supset X_1\supset\dots\supset X_{n-1}=\left(\Gamma_1\times D_{n-1}(\tilde \Gamma_2)\right)\cup D_n(\Gamma_2).
\end{equation}
The sequences are of the form
\begin{equation}\label{m-v_h2arbitrary}
\dots \rightarrow H_3(X_k)\xrightarrow{\delta} H_2(A_k\cap B_k)\xrightarrow{\Phi_k}H_2(A_k)\oplus H_2(B_k)\xrightarrow{\Psi_k}H_2(X_k)\xrightarrow{\delta_k}H_1(A_k\cap B_k)\xrightarrow{}\dots
\end{equation}
There are two major differences in comparison to the previously considered settings. First of all, the third homology group of the subcomplexes does not vanish, hence we do not have any {\it a priori} knowledge about the kernel of $\Phi_k$. Therefore, we shall conjecture that $\im\delta=0$.
\begin{conj}\label{delta_conj}
Consider $D_n(\Gamma)$ for $\Gamma$ consisting of two arbitrary graphs connected by a single edge. Let $X_k$ be any subcomplex in the decomposition of $D_n(\Gamma)$ (see equation (\ref{eq:single_edge_decomp})). We conjecture that the boundary map from the Mayer-Vietoris sequence for $X_k$
\[\delta:\ H_m(X_k)\rightarrow H_{m-1}(A_k\cap B_k),\ m\geq 2.\]
has a trivial image.
\end{conj}
\noindent Map $\delta$ maps the toric part of the over-complete basis of $H_m(D_n(\Gamma))$ to zero. Therefore, the nontrivial elements of $\im\delta$ must come from some more exotic elements of $H_m(D_n(\Gamma))$. Yet, we have not found any numerical evidence for the existence of an  element of $H_m(D_n(\Gamma))$ represented by an $m$-chain that involves a one-connected subgraph with $v$ of valence two. Therefore, we conjecture that for indistinguishable particles on one-connected graphs from such a class, we always have $\im\delta=0$. For a longer discussion with examples of graphs, where $\im\delta\neq0$, see section \ref{sec:partial}. In particular, we show that $\im\delta\neq0$ already for distinguishable particles on tree graphs.

Secondly, the homology groups are in general not free groups. Therefore, the Mayer-Vietoris sequence does not split. However, we will use our conjecture that $\im\delta=0$, which immediately gives $H_2(X_k)\cong\coker\Phi_k$ and $\ker\Phi_k=0$.
\begin{theorem}\label{th:h2_2comp}
Let $\Gamma$ be a wedge sum of $\Gamma_1$ and $\Gamma_2$, where gluing map identifies two vertices of valence $1$ in $\Gamma_1$ and $\Gamma_2$. Moreover, assume that maps $\delta$ and $\delta'$ from the corresponding Mayer-Vietoris sequences for $X_k$ and $X_k'$ satisfy conjecture \ref{delta_conj}. Then, the rank of the second homology group for $n$ indistinguishable particles on $\Gamma$ is 
\begin{equation}\label{h2_2comp}
\beta_2^{(n)}(\Gamma)=\beta_2^{(n)}(\Gamma_2)+\beta_2^{(n)}(\Gamma_1)+\sum_{k=1}^{n}\left(\beta_1^{(k)}(\Gamma_1)-\beta_1^{(k-1)}(\Gamma_1)\right)\beta_{1}^{(n-k)}(\Gamma_2).
\end{equation}
\end{theorem}
\begin{proof}
We prove the theorem in the standard way. The assumption that $\im\delta=0$ and $\im\delta'=0$ implies that maps $\Phi_k$ and $\Phi'_k$ are injective for all $k$. Therefore, the quotients in $\coker\Phi_k$ and $\coker\Phi_k'$ can be again realised as quotients by $\coim\Phi_k$ and $\coim\Phi_k'$ on $H_2(A_k)$ and $H_2(A_k')$ respectively.  Hence, for the Mayer-Vietoris sequence for $X_k'$ we have
\begin{equation}\label{recurrence_h2_2comp}
H_2(B_k')\cong H_2(B_k)\oplus\frac{H_2(D_{n-k}(\Gamma_1))\oplus H_{2}(D_k(\tilde \Gamma_2))\oplus \left(H_1(D_{n-k}(\Gamma_1))\otimes H_{1}(D_k(\tilde \Gamma_2))\right)}{H_2(D_{n-k-1}(\tilde\Gamma_1))\oplus H_{2}(D_k(\tilde \Gamma_2))\oplus\left(H_1(D_{n-k-1}(\tilde \Gamma_1))\otimes H_{1}(D_k(\tilde \Gamma_2))\right)}.
\end{equation}
The sequence for $X_k$ yields $H_2(B_{k-1})\cong H_2(B_k')$. From these equations we obtain the following recurrence relation for the rank of $H_2(B_{k-1})$.
\[\beta_2(B_{k-1})=\beta_2(B_{k})+\beta_2^{(n-k)}(\Gamma_1)-\beta_2^{(n-k-1)}(\Gamma_1)+\left(\beta_1^{(n-k)}(\Gamma_1)-\beta_1^{(n-k-1)}(\Gamma_1)\right)\beta_1^{(k)}(\Gamma_2)\]
for $k=0,1,\dots,n-1$. The initial condition is $\beta_2(B_{n-1})=\beta_2(D_n(\Gamma_2))$. Solving the recurrence and using the fact that $\sum_{k=0}^{n-1}\beta_2^{(n-k)}(\Gamma_1)-\beta_2^{(n-k-1)}(\Gamma_1)=\beta_2(D_n(\Gamma_1))$ we obtain equation (\ref{h2_2comp}).
\end{proof}

\begin{rem}
The reasoning from the proof of the above theorem can be used to describe the torsion of $H_2(D_n(\Gamma))$. Namely, using the fact that $D_k(\Gamma_2)$ is homotopy equivalent to $D_k(\tilde \Gamma_2)$,  recurrence relation (\ref{recurrence_h2_2comp}) can be simplified to the following form.
\[H_2(B_{k-1})\cong H_2(B_k)\oplus\frac{H_2(D_{n-k}(\Gamma_1))}{H_2(D_{n-k-1}(\Gamma_1))}\oplus\left(\frac{H_1(D_{n-k}(\Gamma_1))}{H_1(D_{n-k-1}(\Gamma_1))}\otimes H_{1}(D_k(\Gamma_2))\right),\]
with the initial condition $H_2(B_{n-1})=H_2(D_n(\Gamma_2))$. Hence, the torsion of $H_2(D_n(\Gamma))$ comes from the torsions of $H_2(D_{n-k}(\Gamma_1))/H_2(D_{n-k-1}(\Gamma_1))$ for $k\in\{0,1,\dots,n-2\}$ and the torsion parts of 
\[\left(\frac{H_1(D_{n-k}(\Gamma_1))}{H_1(D_{n-k-1}(\Gamma_1))}\otimes H_{1}(D_k(\Gamma_2))\right)\]
for $k\in\{0,1,\dots,n-1\}$. In particular, the latter yields only copies of $\ZZ_2$, as the only possible torsion of the first homology group is $\ZZ_2$. This is the case if and only if at least one of the components is non-planar \cite{HKRS,KoPark}.
\end{rem}

\section{When is $\im\delta$ nontrivial?}\label{sec:partial}
The main obstacle in the continuation of this paper's approach in a rigorous way is the knowledge of cycles that do not belong to $\ker\delta$. In this section we provide two examples of such cycles for $D_n(\Gamma)$ and $\D_n(\Gamma)$. We conjecture that the types of cycles described in this section are all possible cycles that do not belong to $\ker\delta$. The first class of cycles appears while considering a simultaneous exchange of particles in $D_k(\tilde\Gamma_1)$, $D_l(\tilde\Gamma_2)$ and on a $Y$-subgraph centered at vertex $v$, see Fig.\ref{y_at_v}.
\begin{figure}[ht]
 
\includegraphics[width=0.4\textwidth]{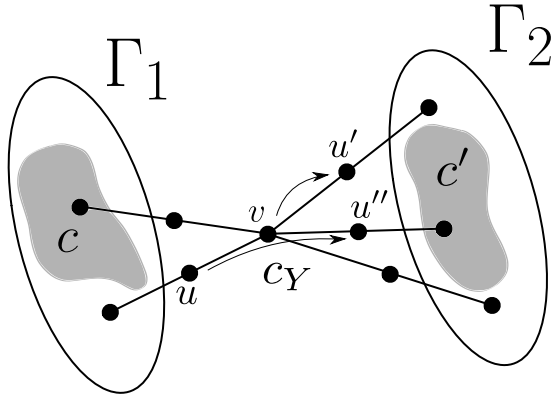}
\caption{A cycle, for which $\delta [c\otimes c'\otimes c_Y]\neq0$ for the boundary map in the proper Mayer-Vietoris seqience. Chain $c_Y$ corresponds to the exchange of two particles on the $Y$-subgraph centered at $v$ and spanned on vertices $u,u',u''$. Chains $c,c'$ are cycles of arbitrary dimensions that are contained in $D_k(\tilde\Gamma_1)$ and $D_l(\tilde\Gamma_2)$ respectively.}
\label{y_at_v}
\end{figure} 
Such a cycle is isomorphic to the tensor product of cycles $z=c\otimes c'\otimes c_Y$, where $c\in\mk{C}_p(D_k(\tilde\Gamma_1))$, $c'\in\mk{C}_s(D_l(\tilde\Gamma_2))$ and 
\[c_Y=\{e_v^{u'},u\}+\{e_u^{v},u'\}+\{e_v^{u''},u'\}-\{e_v^{u'},u''\}-\{e_u^{v},u''\}-\{e_v^{u''},u\}.\]
In our notation, the cycle corresponds to the following boundary map
\[\delta'_{l+1}:\ H_{p+s+1}(X_{l+1})\to H_{p+s}(D_{k+1}(\tilde\Gamma_1)\times D_{l+1}(\tilde\Gamma_2)).\]
The decomposition $z=x+y$, $x\in\mk{C}_{p+s+1}(A_k')$, $y\in\mk{C}_{p+s+1}(B_k')$ yields
\begin{eqnarray*}
x=c\otimes c'\otimes\left(\{e_v^{u'},u\}-\{e_v^{u''},u\}\right)\in\mk{C}_{p+s+1}(D_{k+1}(\tilde\Gamma_1)\times D_{l+1}(\Gamma_2)), \\
y=c\otimes c'\otimes\left(\{e_u^{v},u'\}+\{e_v^{u''},u'\}-\{e_v^{u'},u''\}-\{e_u^{v},u''\}\right)\in\mk{C}_{p+s+1}(D_{k+1}(\Gamma_1)\times D_{l+1}(\tilde\Gamma_2)).
\end{eqnarray*}
For the boundary of a $1$-cell $\partial \{e_a^b,c\}=\{b,c\}-\{a,c\}$, we have $\partial' x=-\partial' y=c\otimes c'\otimes \left(\{u',u\}-\{u'',u\}\right)\neq 0.$

Another class of cycles that do not belong to the kernel of a proper boundary map, describes exchanges of distinguishable particles on pairs of star subgraphs. For simplicity, let us focus on an example of two particles on the double $Y$-graph (Fig.\ref{double_y}).
\begin{figure}[ht]
 
\includegraphics[width=0.3\textwidth]{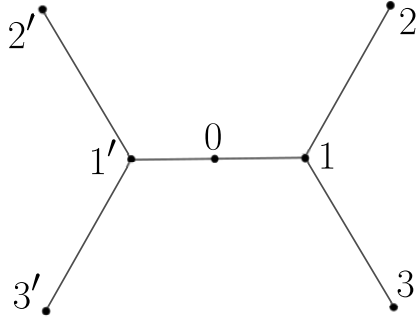}
\caption{Double-$Y$ graph for two particles - wedge sum of two $Y$-graphs, each sufficiently subdivided for two particles.}
\label{double_y}
\end{figure} 
There are two cycles in $\D_2(\Gamma)$ that are analogous to the exchange of indistinguishable particles on a single $Y$-subgraph (see Fig.\ref{double_y_conf}). However, unlike the configuration space for indistinguishable particles, $\D_2(\Gamma)$ contains an additional $1$-cycle that involves both $Y$-subgraphs. This cycle is marked with arrows on Fig.\ref{double_y_conf}. 
\begin{figure}[ht]
 
\includegraphics[width=0.8\textwidth]{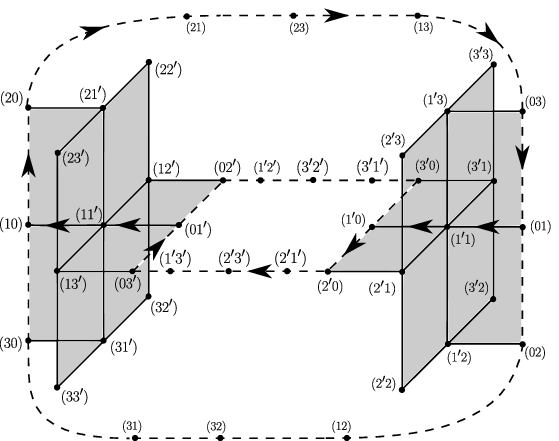}
\caption{Distinguished configuration space $\D_2(\Gamma)$ for $\Gamma$ a double-$Y$ graph form Fig.\ref{double_y}. Cycles marked with dashed line describe exchange of particles on a single $Y$-subgraph, while the cycle marked with arrows describes exchange of particles that involves both $Y$-subgraphs.}
\label{double_y_conf}
\end{figure} 
For the order of vertices $3'<2'<1'<0<1<2<3$ and the boundary map $\partial (e_a^b,c)=(b,c)-(a,c)$, $\partial (c,e_a^b)=(c,b)-(c,a)$, the chain that describes such a cycle is
\begin{eqnarray*}
c_{Y_1,Y_2}=(e_1^2,0)+(2,e_0^1)+(2,e_1^3)-(e_1^2,3)-(e_0^1,3)-(0,e_1^3)-(e_{1'}^0,1)+ \\ -(1',e_0^1)-(e_{2'}^{1'},0)-(2',e_{1'}^0)-(2',e_{3'}^{1'})+(e_{2'}^{1'},3')+(e_{1'}^0,3')+(0,e_{3'}^{1'})+(e_0^1,1')-(1,e_{1'}^0).
\end{eqnarray*}
This is a cycle, which goes through the whole configuration space, hence the proper boundary map for the Mayer-Vietoris sequence is
\[\delta_0:\ H_1(\D_2(\Gamma))\to H_0(\tilde Y_1\times 0)\oplus H_0(0\times \tilde Y_1).\]
The decomposition of $c_{Y_1,Y_2}=x+y$ yields
\[x=-(e_{2'}^{1'},0)-(2',e_{1'}^0)-(2',e_{3'}^{1'})+(e_{2'}^{1'},3')+(e_{1'}^0,3')+(0,e_{3'}^{1'})\in\mk{C}_1(\D_2(Y_1)).\]
Then $\partial_0x=(0,1')-(1',0)\neq0$. We encounter a similar situation while considering a star graph connected with another graph $\Gamma$, where the particles exchange on the star graph and some star subgraph of $\Gamma$. We suppose that for distinguished configuration spaces of tree graphs the cycles described in this section are the only possible cycles that do not belong to $\ker\delta$. However, the proof requires a further study of the discrete Morse theory for distinguishable particles.

\section{Summary}
Without any additional knowledge about the over-complete basis of the homology groups it is hard to continue the above approach in a rigorous way. The main difficulty is a full description of $\im\delta$. However, we suspect that in the case indistinguishable particles the only nontrivial elements from $\im\delta$ are the ones that involve exchanges of particles on $Y$-subgraphs centered at vertex $v$, as described in section \ref{sec:partial}. If one accepts such an assumption, the methods from this paper can be used {\it mutatis mutandis} to express $H_m(D_n(\Gamma))$ of an arbitrary one-connected graph by the homology groups of its higher-connected components. The case of the distinguished configuration spaces is much more difficult, as we point out in section \ref{sec:partial}. 

The methodology developed in this paper can be viewed as one of the possible general frameworks for describing the topology of the configuration spaces of graphs. A generalisation of this approach would be to consider the decomposition of an arbitrary graph to star subgraphs and describe the relations and connections between the corresponding components of the configurations space. This is a subject of our further studies.

\begin{acknowledgments}
We would like to thank Jonathan Robbins for fruitful discussions on the significance of homology groups in quantum theories and Jon Keating for suggesting improvements in the manuscript. TM is supported by Polish Ministry of Science and Higher Education ``Diamentowy Grant'' no. DI2013 016543, European Research Council grant QOLAPS and by the CTP PAS research grant for young researchers. AS would like to thank the Marie Curie International Outgoing Fellowship for financial support.
\end{acknowledgments}


\begin{thebibliography}{99}

\bibitem{L-M} Leinaas, J. M., Myrheim, J., {\it On the theory of identical particles}, Nuovo Cim. 37B, 1-23, 1977

\bibitem{Wilczek} Wilczek, F., {\it Fractional statistics and anyon superconductivity},  Singapore: World Scientific, 1990

\bibitem{Souriau} Souriau, J. M., {\it Structure des systmes dynamiques}, Dunod, Paris 1970 

\bibitem{LD}Laidlaw, M. G. G. and DeWitt, C. M., {\it Feynman functional integrals for systems of indistinguishable particles} Phys. Rev. D 3, 1375-1378, 1971 

\bibitem{Bolte} Bolte J., Kerner J., {\it Quantum graphs with singular two-particle interactions}, J. Phys. A: Math. Theor. 46 (2013) 045206

\bibitem{Dowker} Dowker, J. S. Remarks on non-standard statistics J. Phys. A: Math. Gen. 18 3521, 1985 

\bibitem{Arnold} Arnold, V., I., {\it On some topological invariants of algebraic functions}, Trudy Moscov. Mat. Obshch. 21 1970, 27-46 (Russian), English transl. in Trans. Moscow Math. Soc. 21 1970, 30-52, 1970

\bibitem{Bloore} Bloore F. J., Bratley I. and Selig J. M. {\it SU(n) bundles over the configuration space of three identical particles moving on $\RR^3$}, J. Phys. A: Math. Gen. 16 729, 1983

\bibitem{HKR} Harrison J. M., Keating J. P. and Robbins J. M. {\it  Quantum statistics on graphs} Proc. R. Soc. A 8
January vol. 467 no. 2125 212-233, 2011

\bibitem{HKRS} Harrison, J., M., Keating, J., P., Robbins, J., M., Sawicki A., {\it n-Particle Quantum Statistics on Graphs}, Comm. Math. Phys., Vol. 330, Issue 3, pp 1293-1326, 2014

\bibitem{ASphd} Sawicki, A., {\it Topology of graph configuration spaces and quantum statistics}, PhD thesis, Bristol, 2014

\bibitem{FStree} Farley, D., Sabalka, L., {\it On the cohomology rings of tree braid groups}, J. Pure Appl. Algebra 212 53-71, 2008

\bibitem{FS12} Farley, D., Sabalka, L., {\it Presentations of graph braid groups}, Forum Math. 24 827-859, 2012

\bibitem{FSbraid} Farley, D., Sabalka, L., {\it Discrete Morse theory and graph braid groups}, Algebr. Geom. Topol. 5 1075-1109, 2005

\bibitem{Ghirst} Ghrist, R., {\it Configuration spaces of graphs and robotics}, Braids, Links, and Mapping Class Groups: the Proceedings of Joan Birman's 70th Birthday, AMS/IP Studies in Mathematics, vol. 24, 29-40, 2001

\bibitem{Hatcher} Hatcher, A., {\it Algebraic Topology}, Cambridge University Press, 2002

\bibitem{Farber03} Farber, M., {\it Topological Complexity of Motion Planning}, Discrete and Computational Geometry 29, 211-221, 2003

\bibitem{FG08} Farber, M., Grant, M., {\it Topological complexity of configuration spaces}, Proc. Amer. Math. Soc. 137, 1841-1847, 2009

\bibitem{Farber04} Farber, M., {\it Instabilities of Robot Motion}, Topology and its Applications 140, 245-266, 2004

\bibitem{Farber05} Farber, M., {\it Collision free motion planning on graphs}, Algorithmic Foundations of Robotics IV, Springer,123?138, 2005

\bibitem{AbramsPhD} Abrams, A., {\it Configuration spaces and braid groups of graphs}, Ph.D. thesis, UC Berkley, 2000

\bibitem{KoPark} Ko, K., H., Park, H., W., {\it Characteristics of graph braid groups}, arXiv:1101.2648, 2011

\bibitem{AS12} Sawicki, A., {\it Discrete Morse functions for graph configuration spaces}, J. Phys. A: Math. Theor. 45 505202, 2012

\bibitem{Forman} Forman, R., {\it Morse Theory for Cell Complexes}, Advances in Mathematics 134, 90145, 1998

\bibitem{BF1} Barnett, K., Farber, M., {\it Topology of Configuration Space of Two Particles on a Graph}, arXiv:0903.2180, 2009

\bibitem{BF2} Barnett, K., Farber, M., {\it Topology of Configuration Space of Two Particles on a Graph}, arXiv:0903.2180, 2009

\bibitem{kuratowski} Kuratowski, K., {\it Sur le probl\'{e}me des courbes gauches en topologie}, Fund. Math. 15: 271-283, 1930
\end{thebibliography}
\end{document}